\pdfoutput=1

\documentclass
  [
    DIV=14,
    fontsize=12,
    abstract=true,
    bibliography=totoc
  ]
  {scrartcl}

\setkomafont{title}{\mdseries\rmfamily}
\setkomafont{author}{\large\scshape}
\setkomafont{date}{\small\rmfamily}

\usepackage{authblk}

\newcommand\keywords[1]{%
  \begingroup
  \renewcommand\thefootnote{}
  \footnote{{keywords:} #1}%
  \addtocounter{footnote}{-1}
  \endgroup
}
\newcommand\MSCcodes[1]{%
  \begingroup
  \renewcommand\thefootnote{}
  \footnote{{MSC codes:} #1}%
  \addtocounter{footnote}{-1}
  \endgroup
}
\newcommand\funding[1]{%
  \begingroup
  \renewcommand\thefootnote{}
  \footnote{{funding:} #1}%
  \addtocounter{footnote}{-1}
  \endgroup
}

\NewCommandCopy{\oldparagraph}{\paragraph}
\RenewDocumentCommand{\paragraph}{s O{#3} m}{%
  \IfBooleanTF{#1}
    {\oldparagraph*{#3.}}       
    {\oldparagraph[#2]{#3.}}    
}

\usepackage{caption}
\usepackage
  [leftcaption]
  {sidecap}

\usepackage{enumitem}
\setlist[enumerate,1]{label=\textbf{\textup{(\roman*)}}}
\setlist{
  topsep=1pt,
  itemsep=2pt,
  leftmargin=1cm,
  before=\leavevmode
}

\usepackage[
  backend=biber,
  style=alphabetic,
  giveninits=true, 
  backref=true
]{biblatex}
\usepackage{lmodern}
\usepackage{anyfontsize}
    
\usepackage{multirow}
\usepackage{hhline}

\usepackage{tabularray}
\UseTblrLibrary{booktabs}

\usepackage{adjustbox}
\usepackage{tcolorbox}

\newtcolorbox{standout}{
  colback=gray!15,
  boxrule=0pt,
  left=.3cm,
  right=.3cm,
  top=.18cm,
  bottom=.18cm,
  boxsep=0pt
}

\usepackage{amsmath}
\allowdisplaybreaks
\usepackage{mathtools}
\usepackage{amssymb}
\usepackage{amsthm}
\usepackage{xfrac}
\usepackage{stmaryrd} 
\usepackage{scalerel} 
\usepackage{stackengine} 
\usepackage[safe]{tipa} 
\usepackage[
  cal=euler,
  scr=dutchcal
]
 {mathalpha} 

 \newcommand{\bracket}[3]{%
  \stretchleftright
    {#1}
    {%
      \ensurestackMath{\addstackgap[1pt]{#2}}%
      \vrule width 0pt depth 2pt height 0pt
    }
    {#3}%
} 
\newcommand{\scaledbracket}[3]{%
  \ThisStyle{%
    \stretchleftright
      {#1}
      {
        \ensurestackMath{\addstackgap[1pt]{\SavedStyle #2}}%
        \vrule width 0pt depth 1.5pt height 0pt
      }
      {#3}%
  }%
}

\theoremstyle{plain}
\newtheorem{theorem}{Theorem}[section]
\newtheorem{lemma}[theorem]{Lemma}
\newtheorem{proposition}[theorem]{Proposition}
\newtheorem{corollary}[theorem]{Corollary}
\theoremstyle{definition}

\newtheorem{example}[theorem]{Example}

\theoremstyle{remark}

\usepackage[
  colorlinks=true,
  linkcolor=darkgreen,
  citecolor=darkgreen,
  urlcolor=darkgreen
]{hyperref}
\usepackage{xurl} 

\usepackage{cleveref}
\crefname{equation}{}{}
\crefname{section}{\S}{\S\S}
\crefname{subsection}{\S}{\S\S}
\crefname{subsubsection}{\S}{\S\S}
\crefname{definition}{Def.}{Defs.}
\crefname{theorem}{Thm.}{Thms.}
\crefname{corollary}{Cor.}{Cors.}
\crefname{lemma}{Lem.}{Lems.}
\crefname{proposition}{Prop.}{Props.}
\crefname{remark}{Rem.}{Rems.}
\crefname{notation}{Ntn.}{Ntns.}
\crefname{fact}{Fact}{Fact}
\crefname{example}{Ex.}{Exs.}
\crefname{figure}{Fig.}{Figs.}
\crefname{table}{Tab.}{Tabs.}
\crefname{footnote}{ftn.}{ftns.}
\Crefname{footnote}{Ftn.}{Ftns.}

\usepackage{tikz}
\usetikzlibrary{
  cd,
  calc,
  arrows.meta,
  backgrounds,
  decorations,
  decorations.pathmorphing, 
}

\tikzcdset{
  arrow style=tikz, 
  diagrams={
    trim left=
    {([xshift=5pt]current bounding box.west)},
    trim right=
    {([xshift=-5pt]current bounding box.east)},
    baseline=
    {([yshift=-2pt]current bounding box.center)},
    >={
      Computer Modern Rightarrow[
        length=4pt, width=4pt
      ]
    },
    hook/.style={
      {Hooks[right, length=2pt, width=8pt]}->
    },
    hook'/.style={
      {Hooks[left, length=2pt, width=8pt]}->
    },
    shorten=0pt,
  }
}

\definecolor{darkblue}{rgb}{0.05,0.25,0.65}
\definecolor{darkgreen}{RGB}{20,140,10}
\definecolor{lightgray}{rgb}{0.9,0.9,0.9}
\definecolor{darkorange}{RGB}{200,100,5}
\definecolor{darkyellow}{rgb}{.91,.91,0}
\definecolor{lightolive}{RGB}{225, 220, 185}

\let\originalsslash\sslash
\renewcommand{\sslash}{\mathord{\originalsslash}}

\makeatletter
\newcommand{\cpt}{\mathpalette\cpt@inner\relax}
\newcommand{\cpt@inner}[2]{%
  \scalebox{0.5}[0.9]{$#1\cup$}
  #1\{\infty\}
}
\makeatother

\makeatletter
\newcommand{\plus}{\mathpalette\sqcpt@inner\relax}
\newcommand{\sqcpt@inner}[2]{%
  \scalebox{0.5}[0.9]{$#1\sqcup$}
  #1\{\infty\}
}
\makeatother

\tikzset{
  snake left/.style={
    rounded corners,
    to path={
      let \p1 = (\tikztostart.east),
          \p2 = (\tikztotarget.west),
          \p3 = ($(\p1)!0.5!(\p2)$),
          \n1 = {8pt} 
      in
      (\p1)
      -- (\x1 + \n1, \y1)
      -- (\x1 + \n1, \y3)
      -- (\x2 - \n1, \y3) \tikztonodes
      -- (\x2 - \n1, \y2)
      -- (\p2)
    }
  }
}

\tikzset{
  uphordown/.style={
    rounded corners,
    to path={
      let \p1 = (\tikztostart.north),
          \p2 = (\tikztotarget.north),
          \n1 = {max(\y1,\y2) + 8pt}
      in
      (\p1)
      -- (\x1, \n1)
      -- (\x2, \n1) \tikztonodes 
      -- (\p2)
    }
  }
}

\tikzset{
  downhorup/.style={
    rounded corners,
    to path={
      let \p1 = (\tikztostart.south),
          \p2 = (\tikztotarget.south),
          \n1 = {min(\y1,\y2) - 8pt}
      in
      (\p1)
      -- (\x1, \n1)
      -- (\x2, \n1) \tikztonodes 
      -- (\p2)
    }
  }
}

\tikzset{
  rightvertleft/.style={
    rounded corners,
    to path={
      let \p1 = (\tikztostart.east),
          \p2 = (\tikztotarget.east),
          \n1 = {max(\x1,\x2) + 8pt}
      in
      (\p1)
      -- (\n1, \y1)
      -- (\n1, \y2) \tikztonodes 
      -- (\p2)
    }
  }
}

\tikzset{
  leftvertright/.style={
    rounded corners,
    to path={
      let \p1 = (\tikztostart.west),
          \p2 = (\tikztotarget.west),
          \n1 = {min(\x1,\x2) - 8pt}
      in
      (\p1)
      -- (\n1, \y1)
      -- (\n1, \y2) \tikztonodes 
      -- (\p2)
    }
  }
}

\newcommand{\inlinetikzcd}[1]{\begin{tikzcd}[sep=small, ampersand replacement=\&]#1\end{tikzcd}}

\newcommand{\defneq}{\equiv}

\newcommand{\ClassifyingA}{\mathcal{A}}
\newcommand{\ClassifyingB}{\mathcal{B}}

\newcommand
{\BHatSUTwoFibration}
{p}

\newcommand
{\HatESUTwoFibration}
{\wp}

\newcommand
{\NameEmail}[2]
{%
  \stackunder%
  {#1}%
  {%
   \clap{%
   \normalfont%
   \footnotesize%
   \texttt{#2}%
  }}%
  \hspace{-3.5pt}%
}

\begin{document}

\setlength{\abovedisplayskip}{3.5pt}
\setlength{\belowdisplayskip}{3.5pt}
\setlength{\abovedisplayshortskip}{-5pt}
\setlength{\belowdisplayshortskip}{3pt}

\title
{Flux Quantization of Type IIA
\\
in Unstable K-Theory}

\author[1]
{\NameEmail 
  {Pinak Banerjee}
  {pinakb24@vt.edu}
}

\author[2,3]
{\NameEmail
  {Hisham Sati}
  {hsati@nyu.edu}
}

\author[2]
{\NameEmail
  {Urs Schreiber}
  {us13@nyu.edu}
}

\affil[1]{Department of Physics, Virginia Tech, USA}
\affil[2]{Mathematics Program and Center for Quantum and Topological Systems (CQTS),  \newline New York University Abu Dhabi (NYUAD), UAE}
\affil[3]{The Courant Institute for Mathematical Sciences, New York University, New York, USA}

\maketitle


\MSCcodes{
  Primary:
  8E50,
  81T30,
  19L50;
  Secondary:
  55N20,
  55R35,
  55P62
}

\funding{
  by \textit{Tamkeen UAE} under the \textit{Abu Dhabi Research Institute Grant} \texttt{CG008}
}

\vspace{-2.5cm}

\begin{abstract}
  The traditional conjecture that RR-fluxes are quantized in stable K-cohomology fails to account for the presence of NS-brane sources: These impose nonlinear relations --- reductions of the famous quadratic relation on M-brane flux --- that can only be captured by unstable nonabelian cohomology theories. Here we consider a deformation of unstable K-theory which properly quantizes the fluxes coupling to D0/D2/NS5-branes, find a twisted version that quantizes also the fluxes coupling to NS1/D4-branes, and show that this oxidizes to a proper electromagnetic quantization of M-brane fluxes.
\end{abstract}

\medskip

\begin{center}
\begin{minipage}{9cm}
\tableofcontents
\end{minipage}
\end{center}

\newpage

\section
{The Flux Quantization}
\label
{OnTheChargeQuantization}

\paragraph
{The Nonlinearity in the IIA Bianchis}
The NS/RR-flux densities, 
$H_3$, $H_7 := \star H_3$ and $F_{2\bullet} = \star F_{10-2\bullet}$,
in duality-symmetric massive type IIA 10D supergravity are subject to the following Bianchi identities (cf. \cite[\S 22.1.3]{Ortin2015} after \parencites{GianiPernici1984}{HuqNamazie1985}{DallAgataEtAl1998}[\S 3]{CremmerJuliaLuPope1998}):
\begin{subequations}
\label{TypeIIABianchis}
  \begin{align}
    \label{TheF2BulletBianchi}
    \mathrm{d}\, F_{2\bullet}
    & =
    H_3 \wedge F_{2\bullet - 2}
    \\
    \label{TheH3Bianchi}
    \mathrm{d}\, H_3 
      & = 0
    \\
    \label{TheH7Bianchi}
    \mathrm{d}\, H_7 
      & = 
      \tfrac{1}{2} F_4 
        \wedge
      F_4
      -
      F_2 \wedge F_6
      + 
      F_0 \wedge F_8
      \mathrlap{\,.}
  \end{align}
\end{subequations}

In discussing possible flux/charge quantization laws \parencites{Freed2000}{SS25-Flux} for these relations, it has been tradition to ignore the last one \cref{TheH7Bianchi}. Granting that, the first relation \cref{TheF2BulletBianchi} has the form characterizing the Chern character image in twisted complex topological K-theory (cf. \parencites[\S 6.3]{BouwknegtEtAl2002}{MathaiStevenson2003}), with the twist being in ordinary 3-cohomology quantizing the $H_3$-flux in \cref{TheH3Bianchi}. 

This is essentially the origin of the famous conjecture that D/NS-brane charge is quantized in 3-twisted topological K-theory (\cite{GreenHarveyMoore1996,MinasianMoore1997,MooreWitten2000}, cf. \cite[\S 4.1]{SS25-Flux}). A lot of work has been done under this assumption (cf. \cite{Witten1998,DistlerFreedMoore2009,GrS22-RRFields,HosseiniTachikawaZhang2025}).

However, the last relation \cref{TheH7Bianchi} cannot actually be disregarded if one really means to speak about type IIA supergravity.
Notice in particular that: 
\begin{enumerate}
\item the standard K-theory conjecture means to quantize not only the magnetic ($F_{< 5}$) but also the electric ($F_{> 5}$) RR-flux densities, which highlights that there is no reason to disregard quantization of the electric NS flux $H_7$, much less ignore the relation it imposes on the other fluxes already rationally, even if not itself quantized;

\item
the nonlinear $H_7$ Bianchi identity \cref{TheH7Bianchi} is (cf. \parencites[\S 4.2]{MathaiSati2004}[\S 3]{FSS17-Sphere}) the direct image under dimensional reduction (cf. \cite[\S 22.1]{Ortin2015}) of the more widely appreciated nonlinearity in the Bianchi identity on the flux densities, $G_4$ and $G_7 := \star G_4$, in 11D supergravity (cf. \parencites[\S 3.1.3]{MiemiecSchnakenburg2006}[Thm. 3.1]{GSS24-SuGra} after \cite{CremmerJuliaScherk1978}):
\begin{subequations}
\label{TheCFieldBianchis}
\begin{align}
  \mathrm{d}\, G_4 & = 0
  \mathrlap{\,,}
  \\
  \label{TheG7Bianchi}
  \mathrm{d}\, G_7 & = \tfrac{1}{2} G_4 \wedge G_4
  \mathrlap{\,.}
\end{align}
\end{subequations}
\end{enumerate}
Hence if the global completion of 11D/10D supergravity by flux quantization is to be compatible with M/IIA duality, it must account for these nonlinearities \cref{TheH7Bianchi,TheG7Bianchi}.

\paragraph
{Proper Electromagnetic Flux Quantization}
In general, flux/charge quantization means \cite{SS25-Flux,SS24-Phase} to choose a generalized cohomology theory whose \emph{character map} \cite{FSS23-Char} image reproduces the given duality-symmetric Bianchi identities, and then to ask for charge preimages of the flux densities under this map, cf. \cref{FluxQuantizationSchematics} and \cref{GaussLawFromRelations}. 

Specifically, this applies to the spatial components of the flux densities on any Cauchy surface $X^d$ of spacetime $X^{1,d}$, where the Hodge duality partners become independent of each other and the Bianchi identities become the Gauss law constraints, cf. \cite[Thm. 2.2]{SS24-Phase}:

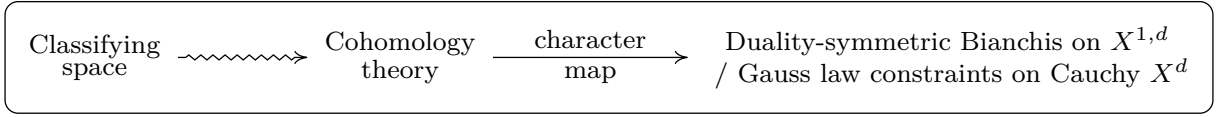
\begin{figure}[htb]
\caption{\label{FluxQuantizationSchematics}
  The \emph{flux/charge quantization} of higher gauge fields means to choose a generalized cohomology theory whose character map takes values in differential forms satisfying the given duality-symmetric Bianchi identities on spacetime $X^{1,3}$, or equivalently the Gauss law constraints on a Cauchy surface $X^d$. After this choice, the globally (topologically) completed field content consists, besides these flux densities, of a choice of cocycle in their essential character preimage.
}
\centering
\adjustbox{
  rndfbox=4pt,
  scale=1.23
}{
$
  \begin{tikzcd}[column sep=40pt]
    \substack{
      \text{Classifying}
      \\
      \text{space}
    }
    \arrow[r, rightsquigarrow] 
    &
    \substack{
      \text{Cohomology} 
      \\
      \text{theory}
    }
    \ar[
      r,
      "{
        \substack{
          \text{character}
        }
      }",
      "{
        \substack{
          \text{map}
        }
      }"{swap}
    ]
    &[20pt]
    \substack{
    \text{Duality-symmetric Bianchis on $X^{1,d}$}
    \\
    \text{ / Gauss law constraints on Cauchy $X^d$} 
    }
  \end{tikzcd}
$
}
\end{figure}

Traditionally this has been considered (cf. \cite{Freed2000}) for Whitehead-generalized abelian cohomology theories like ordinary cohomology, K-theory, elliptic cohomology, etc. In their twisted generalization these abelian cohomology theories capture \emph{linear} Bianchi identities like \cref{TheF2BulletBianchi} ---  but abelian cohomology theories can never (cf. \cite[Ex. 5.6]{FSS23-Char}) capture nonlinear Bianchi identities like \cref{TheH7Bianchi,TheG7Bianchi}.

Concretely in the case \cref{TypeIIABianchis} at hand, the preimages of the RR-flux densities $F_{2\bullet > 0}$ under the (Chern) character map are polynomials $\mathrm{ch}_\bullet$ in the \emph{Chern classes} $c_\bullet \in H^{2\bullet}\bracket({ \mathrm{KU}_0; \mathbb{Z} })$ on the stable classifying space for complex K-theory (cf. \parencites[\S V.3]{Karoubi1978}[\S I.2]{FreedEtAl2008}[Ex. 7.2]{FSS23-Char}):
\begin{equation}
\label{StableBU}
  \mathrm{KU}_0
  \simeq
  B \mathrm{U}(\infty) 
  \times
  \mathbb{Z}
    := 
  \bigcup_{n} B \mathrm{U}(n)
  \times
  \mathbb{Z}
  \mathrlap{\,.}
\end{equation}

Despite its superficial appearance, this space \cref{StableBU} is famously an infinite loop space (by homotopy-theoretic \emph{Bott periodicity}, cf. \parencites[\S B.2]{AguilarGitlerPrieto2002}[\S 15]{HusemollerEtAl2008}), which makes complex K-theory be an abelian cohomology theory (cf. \cite[Ex. 2.10]{FSS23-Char}). 

Hence if we are to account for the nonlinearity in \cref{TheH7Bianchi}, then we have to deform this space and break its infinite loop space property to obtain a classifying space of an unstable \emph{nonabelian cohomology} theory \cref{NonabelianCohomology} \parencites[Def. 6.0.6]{Toen2002}[Def. 6]{Lurie2014}[\S 2]{FSS23-Char} (more exposition in \parencites[\S 1]{SS25-TEC}[\S 4]{SS26-Orb}).

\paragraph
{D0/D2/NS5-sector and Unstable K-Theory}
To warm up to the idea of deforming this situation to account for the nonlinear Bianchi \cref{TheH7Bianchi}, consider first  the special case when only $F_2$, $F_4$ and $H_7$ (the fluxes coupling to D0/D2/NS5-branes) are non-zero (understanding all flux densities from now on as their spatial components on a Cauchy surface $X^d$). In this case, \cref{TypeIIABianchis} reduces to:
\begin{subequations}
\label{TheSimplifiedBianchis}
\begin{align}
  \label{BianchiOnF2F4}
  \mathrm{d}\, F_2 & = 0
  \,,\;
  \mathrm{d}\, F_4  = 0
  \mathrlap{\,,}
  \\
  \label{SimplifiedH7Bianchi}
  \mathrm{d}\, H_7 
    & = \tfrac{1}{2} F_4 \wedge F_4
    \mathrlap{\,.}
\end{align}
\end{subequations}
Taken at face value, the first relation \cref{BianchiOnF2F4} is flux-quantized by the differential nonabelian cohomology theory whose classifying space is the finite stage $B \mathrm{U}(2)$ of \cref{StableBU}. This may be understood as a form of differential \emph{unstable K-theory} (cf. \parencites[Ex. 1.18]{SS25-Complete}[Prop. 9.4]{FSS23-Char}), following the terminology of ``unstable K-theory'' from \parencites[\S 1.2]{vanderKallen1980}, cf. \parencites{HamanakaKono2003}{Hamanaka2003}{HamanakaKono2004}[\S 1.3]{ClausenJansen2024}.

For comparison: On globally hyperbolic 10D spacetimes, the stable range for unstable K-theory is that classified by $B \mathrm{U}(n > 4)$. Hence while $B \mathrm{U}(2)$ is ``very unstable'' in absolute terms, for 10D spacetimes it is halfway through the unstable range.

To obtain from this a classifying space $B \widehat {\mathrm{U}(2)}$ whose induced flux quantization is compatible also with the nonlinearity \cref{SimplifiedH7Bianchi}, the canonical step is to pass to the \emph{homotopy fiber} of the map classifying $(c_2)^2 := c_2 \cup c_2$, making a homotopy fiber sequence \cref{ClassifyingMapForC2Square}:
\begin{equation}
\label
{TheSpaceBHatU2}
  \begin{tikzcd}
      B 
      \widehat{\mathrm{U}(2)}
    \ar[
      rr,
      "{
        \mathrm{hofib}\scaledbracket({
          (c_2)^2
        })
      }"
    ]
    &&
    B \mathrm{U}(2)
    \ar[r, "{ (c_2)^2 }"]
    &
    K\bracket({\mathbb{Z},8})
    \mathrlap{\,.}
  \end{tikzcd}
\end{equation}
(Here $\widehat{\mathrm{U}(2)}$ itself may be understood as a Lie 6-group which as such is a cousin of the \emph{MFivebrane 6-group} of \cite[Ex. 3.2]{FSS21-Hopf}, following terminology of \cite{SSS12,SSSt09-Fivebrane}.)

Hence the nonabelian differential cohomology theory (cf. \cite[\S\S 2,9]{FSS23-Char}) classified by this $B \widehat{\mathrm{U}(2)}$ --- being an admissible flux quantization law for \cref{TheSimplifiedBianchis}, cf.  \cref{MinSullivanModelOfBhatU2} --- is a deformation of unstable topological K-theory that properly takes into account the nonlinear Bianchi identity \cref{TheH7Bianchi} for the NS flux $H_7$ in the case that $F_6 = 0$. 

\paragraph
{Proper Flux Quantization of Type IIA and Twisted Unstable K-theory}
After this warmup exercise for going beyond abelian cohomology theories, we next ask for a comprehensive flux quantization of \cref{TypeIIABianchis} and a compatible lift to M-theory quantizing also \cref{TheCFieldBianchis}.
On this second point, we highlight (with \cite{BMSS2019}) the open secret that the derivation of massive type IIA from 11D is elusive, in that what one really gets by double dimensional reduction of the 11D SuGra fluxes \cref{TheCFieldBianchis} is (cf. again \parencites[\S 4.2]{MathaiSati2004}[\S 3]{FSS17-Sphere})
only the restriction of \cref{TypeIIABianchis} to:
\begin{enumerate}
\item $F_2$ (the class of the 11D circle bundle), 
\item $F_4$ (the direct image of $G_4$),
\item $F_6$ (the reduced image of $G_7$),
\item $H_3$ (the reduced image of $G_4$)
\item $H_7$ (the direct image of $G_7$):
\end{enumerate}
\begin{subequations}
\label
{TheActualBianchisIn10D}
\begin{align}
  \mathrm{d}
  \, 
  F_2 & = 0
  ,\;
  \mathrm{d}\, F_4 
   = H_3 \wedge F_2
  ,\;
  \mathrm{d}\, F_6 
  = 
    H_3 \wedge F_4
  \\
  \mathrm{d}\, H_3 & = 0
  \\
  \mathrm{d}\, H_7 & =
  \tfrac{1}{2}
  F_4 \wedge F_4 
    - 
  F_2 \wedge F_6
  \mathrlap{\,.}
\end{align}
\end{subequations}

This being the exact reduction of the 11D Bianchis \cref{TheCFieldBianchis}, we may appeal to the following theorem (cf. \parencites[\S 2.2]{BMSS2019}{SS24-Cyc}, going back to \cite[\S 3]{FSS18-TD}): For $\ClassifyingB$ any classifying space that quantizes the fluxes in 11D, its \emph{cyclic loop space}
\begin{equation}
\label
{Cyclification}
  \mathrm{Cyc}\bracket({
    \ClassifyingB
  })
  :=
  \bracket({L\ClassifyingB})
  \sslash
  S^1
\end{equation}
is an admissible flux quantization for the resulting Bianchis \cref{TheActualBianchisIn10D}.

But we show in \cref{OnTheProofs} (\cref{HatBSU2AdmissibleClassifyingSpace}) that an admissible choice for $\ClassifyingB$ in 11D is 
\begin{equation}
\label
{TheSpaceHatBSU2}
  B \widehat{\mathrm{SU}(2)}
  :=
  \mathrm{hofib}\Big(\! {
    \begin{tikzcd}
      B \mathrm{SU}(2)
      \ar[
        r,
        "{ (c_2)^2 }"
      ]
      &
      K\bracket({
        \mathbb{Z},
        8
      })
    \end{tikzcd}
  } \!\Big)
  \mathrlap{.}
\end{equation}
This implies (\cref{RationalModelOfCyclifiedSpace}) that the space 
\begin{equation}
\label
{TheCyclifiedClassifyingSpace}
  \mathrm{Cyc}\bracket({
    B \widehat{\mathrm{SU}(2)}
  })
\end{equation}
classifies an admissible flux quantization for the type IIA fluxes \cref{TheActualBianchisIn10D}. 

The universal property of the $\mathrm{Cyc}$-construction then implies (\cref{TheExtCycAdjunct}) that there is a canonical comparison map:
\begin{equation}
\label
{TheAdjunctMapIntoCycBHatSU2}
  \begin{tikzcd}
    B \widehat{\mathrm{U}(2)}
    \ar[r]
    &
    \mathrm{Cyc}\bracket({
      B \widehat{\mathrm{SU}(2)}
    })
    \mathrlap{\,,}
  \end{tikzcd}
\end{equation}
which reflects (\cref{PullbackAlongExtCycAdjunct}) exactly the lift to flux-quantized fields of the inclusion of the restricted situation \cref{BianchiOnF2F4} into the full situation \cref{TheActualBianchisIn10D}. In view of the previous discussion, we may hence think of \cref{TheCyclifiedClassifyingSpace} as classifying a deformed and \emph{unstable} version of \emph{3-twisted K-theory}.

This shows how the traditional conjecture, that D-brane charge is quantized in stable K-theory, may be corrected to account for traditionally ignored nonlinearities in the Bianchi identities, by passage to a deformed version of unstable K-theory, and how after this correction the flux-quantized situation is actually a reduction of a corresponding situation in 11D, as it should be.

\paragraph
{Relation to Hypothesis H}
Beware here that in saying so we are using a slightly different flux quantization in 11D than was considered in \cite{FSS20-H,FSS21-Hopf} following \cite[\S 2.5]{Sati2018}, where the classifying space in 11D was instead taken to be $S^4$, quantizing the C-field fluxes in 4-Cohomotopy (``Hypothesis H''). 

But the two are closely related and in fact indistinguishable in compactifications to $\mathbb{R}^{1,3}$: There is a canonical comparison map
\begin{equation}
\label{The7Equivalence}
  \begin{tikzcd}
    S^4
    \ar[
      r,
      "{
        \sim_{\leq 7}
      }"
    ]
    &
    B \widehat{\mathrm{SU}(2)}
    \mathrlap{\,,}
  \end{tikzcd}
\end{equation}
which is an integral homotopy equivalence up to dimension 7 (this is our main \cref{ComparisonMapIs7Equivalence}),
and hence in particular a rational homotopy equivalence (\cref{ComparisonMapIsRationalEquivalence}).

This highlights once more that the global completion of higher gauge theories by flux quantization is a choice which needs to be justified by checking that it implies desired or expected consequences.

The choice $\ClassifyingB \defneq S^4$ stands out in that it provably implies core topological effects expected in ``M-theory'' (cf. \cite{Duff1999World}), such as the $\tfrac{1}{4}p_1$-shifted flux quantization of $G_4$ \cite[Prop. 3.13]{FSS20-H} and anomaly cancellations of the M5-brane \parencites[Thm. 4.8]{FSS21-Hopf}{SS21-M5Anomaly}. 
These results are unlikely to carry over to $\ClassifyingB \defneq B \widehat{\mathrm{SU}(2)}$, notably in that the relevant tangential twisting by $\mathrm{Spin}(5)$-structure no longer applies. 

However, we may turn this around: The 7-equivalence \cref{The7Equivalence} shows that when the tangential twist is trivial (such as for compactification on tori or other group manifolds) and the compactified space is of dimension $\leq 7$, then the cohomotopical flux quantization with $\ClassifyingB \defneq S^4$ becomes equivalent to that with $\ClassifyingB \defneq B \widehat{\mathrm{SU}(2)}$ and as such the deformed unstable K-theoretic nature of its dimensional reduction follows. 

These equivalences in low dimension may hence serve to relate different hypotheses about flux quantization to each other. Compare also the fact that the stabilization (homotopy linearization) of $S^4$ is equivalent to the classifying space for \emph{topological modular forms} ($\mathrm{tmf}$) in degree 4, over 10-manifolds (\parencites[\S 4.3]{Hopkins2002}[Ex. 7.5]{FSS23-Char}):
\begin{equation}
  \begin{tikzcd}
    \Sigma^\infty S^4
    \ar[
      r,
      "{ \sim_{\leq 10} }"
    ]
    &
    \Sigma^4 \mathrm{tmf}
    \mathrlap{\,.}
  \end{tikzcd}
\end{equation}

\paragraph
{Including Worldvolume Flux on D4/M5}

More generally, in the presence of M5-brane probes (\parencites{HoweSezgin1997}{HoweSezginWest1997}{AganagicEtAl1997}[\S 5.2]{Sorokin2000}[\S 3]{GSS25-M5}) of the 11D SuGra bulk along an embedding
\begin{equation}
 \begin{tikzcd}
   \Sigma^{1,5}
   \ar[
     r,
     hook,
     "{ \phi }"
   ]
   &
   X^{1,10}
   \mathrlap{\,,}
 \end{tikzcd}
\end{equation}
there famously is a (nonlinearly) self-dual flux density $\mathscr{H}_3$ on the worldvolume $\Sigma^{1,5}$ (cf. \cite[Rem. 3.19]{GSS25-M5}) satisfying, on top of \cref{TheCFieldBianchis}, the further worldvolume Bianchi identity (cf. \cite[Prop. 3.18]{GSS25-M5}):
\begin{equation}
\label
{BianchiOnM5Probe}
  \mathrm{d}\, \mathscr{H}_3
  =
  \phi^\ast G_4
  \mathrlap{\,.}
\end{equation}
Equivalently, after double dimensional reduction \cite{DuffHoweInamiStelle1987} of the M5 to a D4-probe of type IIA (cf. \parencites{Townsend1996}[\S 6]{AganagicEtAl1997}[\S 6]{AganagicEtAl1997b}), 
\begin{equation}
\begin{tikzcd}
  \Sigma^{1,4}
  \ar[
    r,
    hook,
    "{ \varphi }"
  ]
  &
  X^{1,9}
  \mathrlap{\,,}
\end{tikzcd}
\end{equation}
this becomes (cf. \parencites{Stefanski2001}[\S 3]{Zhou2001}[\S 2]{Banerjee2026M5brane}):
\begin{subequations}
\label
{TheActualBianchisOnD4}
\begin{align}
  \mathrm{d}\, \mathscr{F}_2 
  & = 
  \varphi^\ast H_3
  \\
  \mathrm{d}\, \mathscr{H}_3
  & =
  \varphi^\ast F_4 
    -
  \varphi^\ast F_2 
    \wedge
  \mathscr{F}_2
  \mathrlap{\,,}
\end{align}
\end{subequations}
where $\mathscr{F}_2$ is the twisted flux density of the D4-worldvolume Chan-Paton gauge field and $\mathscr{H}_3 = \star_{_{1,4}} \mathscr{F_2}$ its dual electric flux density.

Flux quantization of such \emph{twisted relative} Bianchi identities (namely twisted by the bulk fluxes relative to the worldvolume embedding) happens in \emph{twisted relative nonabelian cohomology theories} (\parencites[\S 3.2]{BaSS26-MString}[\S 2]{SS26-Orb}[\S 3]{FSS23-Char}), where classifying spaces $\ClassifyingB$ are enhanced to classifying fibrations $\inlinetikzcd{ \ClassifyingA \ar[r, ->>] \& \ClassifyingB }$:
\begin{equation}
  H^0\bracket({
    \phi;
    \wp
  })
  :=
  \pi_0
  \,
  \mathrm{Map}\left(
    \begin{tikzcd}[
      column sep=15pt
    ]
      \Sigma^{1,5}
      \ar[r, dashed]
      \ar[
        d,
        hook,
        "{ \phi }"
      ]
      &
      \ClassifyingA
      \ar[
        d,
        ->>,
        "{ \wp }"
      ]
      \\
      X^{1,10}
      \ar[
        r,
        dashed
      ]
      &
      \ClassifyingB
    \end{tikzcd}
  \right)
  \mathrlap{.}
\end{equation}

The minimal (in number of CW-cells) admissible choice of flux quantization $\wp$ for higher gauge sector of M5@11Bulk turns out to be (\parencites[\S 3.7]{FSS20-H}{FSS21-Hopf}{FSS21-StrStruc}) the quaternionic Hopf fibration \cref{TheQuaternionicHopfFibration} $ \inlinetikzcd{ S^7 \ar[r, "{ h_{\mathbb{H}} }"] \& S^4 }$, which gives the flux quantization of 11D SuGra with M5-probes according to \emph{Hypothesis H}.

Hence, in view of the above discussion, we are now to ask whether the comparison map \cref{The7Equivalence} to the alternative K-flavored flux quantization law extends to a 7-equivalence of classifying fibrations. 
Indeed, we find a canonical $\mathrm{SU}(2)$-fiber bundle \cref{FibrationOfHattedSpaces} over $B \widehat{\mathrm{SU}(2)}$ equipped with compatible comparison maps \cref{ComparisonMapFromS7} which are integral homotopy equivalences in dimensions $\leq 7$ (\cref{ComparisonOfFibrationsIs7Equivalence}) and hence are in particular rational homotopy equivalences (\cref{FiberComparisonIsRationalEquivalence}):
\begin{equation}
  \begin{tikzcd}[row sep=15pt, column sep=large]
    S^7
    \ar[
      r,
      "{
        \sim_{\leq 7}
      }"
    ]
    \ar[
      d,
      ->>,
      "{ h_{\mathbb{H}} }"
    ]
    &
    B^7 \mathbb{Z}
    \ar[
      d,
      ->>,
      "{ \wp }"
    ]
    \\
    S^4
    \ar[
      r,
      "{
        \sim_{\leq 7}
      }"
    ]
    &
    B \widehat{\mathrm{SU}(2)}
    \mathrlap{\,.}
  \end{tikzcd}
\end{equation}

\begin{figure}[htb]
\caption{\label{Overview}
  \textbf{Overview for IIA.} On the bottom left the bulk type IIA Bianchi identities \cref{TheActualBianchisIn10D} that appear via double dimensional reduction from 11D. On the top left the further Bianchis \cref{TheActualBianchisOnD4} on a D4-probe embedding $\varphi$. On the right the corresponding classifying maps of admissible flux quantization of this situation, first according to Hypothesis H and, on the far right, according to the more K-flavored flux quantization law established here.
  This dimensionally oxidizes to flux quantizations of 11D SuGra surveyed in \cref{Overview11D}.
}
\centering
\adjustbox{
  rndfbox=4pt
}{
$
  \begin{aligned}
  \mathrm{d}\, \mathscr{F}_2 
  & = 
  \varphi^\ast H_3
  \\
  \mathrm{d}\, \mathscr{H}_3
  & =
  \varphi^\ast F_4 
    -
  \varphi^\ast F_2 
    \wedge
  \mathscr{F}_2
  \\[+5pt]
  \mathrm{d}
  \, 
  F_2 & = 0
  \\[-2pt]
  \mathrm{d}\, F_4 
   &= H_3 \wedge F_2
  ,\;
  \mathrm{d}\, F_6 
  = 
    H_3 \wedge F_4
  \\[-2pt]
  \mathrm{d}\, H_3 & = 0
  \\[-2pt]
  \mathrm{d}\, H_7 & =
  \tfrac{1}{2}
  F_4 \wedge F_4 
    - 
  F_2 \wedge F_6
  \mathrlap{\,.}
  \end{aligned}
  \hspace{.8cm}
  \begin{tikzcd}[
    row sep=45pt
  ]
    \Sigma^{1,4}
    \ar[
      d,
      hook,
      "{ \varphi }"
    ]
    \ar[
      r,
      dashed
    ]
    &
    \mathrm{Cyc}\bracket({
      S^7
    })
    \ar[
      r,
      "{
         \mathrm{Cyc}\bracket({
           \hat h_{\mathbb{H}}
         })
      }"
    ]
    \ar[
      d,
      ->>,
      "{ 
        \mathrm{Cyc}(h_{\mathbb{H}}) 
      }"{description}
    ]
    &
    \mathrm{Cyc}\bracket({
      B^7 \mathbb{Z}
    })
    \ar[
      d,
      ->>,
      "{
        \mathrm{Cyc}({
          \wp
        })
      }"{description}
    ]
    \\
    X^{1,9}
    \ar[
      r,
      dashed
    ]
    &
    \mathrm{Cyc}\bracket({
      S^4
    })
    \ar[
      r,
      "{
        \mathrm{Cyc}\bracket({
          \hat 1
        })
      }"
    ]
    &
    \mathrm{Cyc}\bracket({
      B \widehat{\mathrm{SU}(2)}
    })
  \end{tikzcd}
$
}
\end{figure}

As before, by double dimensional reduction of flux quantization laws via cyclification \cref{Cyclification} of classifying spaces \parencites[\S 2.2]{BMSS2019}, it follows at once that the cyclified fibration
\begin{equation}
  \begin{tikzcd}[row sep=15pt]
    \mathrm{Cyc}\bracket({
      B^7 \mathbb{Z}
    })
    \ar[
      d,
      ->>,
      "{
        \mathrm{Cyc}(\wp)
      }"
    ]
    \\
    \mathrm{Cyc}\bracket({
      B 
      \widehat{\mathrm{SU}(2)}
    })
  \end{tikzcd}
\end{equation}
classifies a twisted relative flux quantization law for type IIA with D4-brane probes (cf. \cite[\S 6]{AganagicEtAl1997b}). The situation is summarized in \cref{Overview}.

\begin{figure}[htb]
\caption{\label{Overview11D}
  \textbf{Overview for 11D.} On the bottom left the Bianchi identities in the bulk of 11D SuGra \cref{TheCFieldBianchis}; on the top left the further Bianchi on an M5 probe \cref{BianchiOnM5Probe}. On the right the corresponding classifying maps of admissible flux quantization of this situation, first according to Hypothesis H and then, on the far right, according to the more K-flavored flux quantization law established here. Under double dimensional reduction this yields exactly the situation in \cref{Overview}.
}
\centering
\adjustbox{
  rndfbox=4pt
}{
$
  \hspace{.5cm}
  \begin{aligned}
    \mathrm{d}\, \mathscr{H}_3
    & =
    \phi^\ast G_4
    \\[15pt]
    \mathrm{d}\, 
    G_4 & = 0
    \\
    \mathrm{d}\, G_7
    & = 
    \tfrac{1}{2}
    G_4 \wedge G_4
  \end{aligned}
  \hspace{1.6cm}
  \begin{tikzcd}[
    row sep=25pt
  ]
    \Sigma^{1,5}
    \ar[
      d,
      "{ \phi }"
    ]
    \ar[
      r,
      dashed
    ]
    &
    S^7
    \ar[
      r,
      "{
        \hat h_{\mathbb{H}}
      }"
    ]
    \ar[
      d,
      "{
        h_{\mathbb{H}}
      }"{description}
    ]
    &
    B^7 \mathbb{Z}
    \ar[
      d,
      "{
        \wp
      }"{description}
    ]
    \\
    X^{1,10}
    \ar[
      r,
      dashed
    ]
    &
    S^4
    \ar[
      r,
      "{ \hat 1 }"
    ]
    &
    B\widehat{\mathrm{SU}(2)}
  \end{tikzcd}
  \hspace{.5cm}
$
}
\end{figure}

\section
{The Proofs}
\label{OnTheProofs}

We freely use tools from algebraic topology and homotopy theory, as found in textbooks such as \cite{Whitehead1978,James1984,Massey1991,Kochman1996,Hatcher2002,AguilarGitlerPrieto2002,Arkowitz2011,Fomenko2016}. In particular we use the homotopy long exact sequences (LES) induced by homotopy fibrations (cf. \cite[\S 9.8]{Fomenko2016}); for concise review in our context cf. \cite[\S A.3]{SS26-BBC}.

\paragraph
{Eilenberg-MacLane spaces.}
For $A$ a discrete abelian group and $n \in \mathbb{N}$, we write $B^n A := K\bracket({A,n})$ for the \emph{Eilenberg-MacLane space} of $A$ in degree $n$ (cf. \cite[\S 6]{AguilarGitlerPrieto2002}), characterized uniquely (up to weak homotopy equivalence) by the property
\begin{equation}
\label
{HomotopyGroupsOfEMSpace}
  \pi_k\, B^n A
  =
  \begin{cases}
    A & \text{if }  k = n
    \\
    \text{trivial} & \text{otherwise}
    \,.
  \end{cases}
\end{equation}
This classifies ordinary cohomology (cf. \parencites[\S 7.1]{AguilarGitlerPrieto2002}[Ex. 2.1]{FSS23-Char}), in that cohomology classes correspond to homotopy classes of classifying maps:
\begin{equation}
\label
{OrdinaryCohomologyClassified}
  H^n\bracket({
    X;
    A
  })
  \simeq
  \pi_0
  \,
  \mathrm{Map}\bracket({
    X,
    B^n A
  })
  \mathrlap{\,.}
\end{equation}

This classical fact is the archetype for the definition of generalized \emph{nonabelian cohomology} \parencites[Def. 6]{Lurie2014}[\S 2]{FSS23-Char}, defined for \emph{any} (pointed) space $\ClassifyingA$ (its \emph{classifying space}) as
\begin{equation}
\label
{NonabelianCohomology}
  H^0\bracket({
    X;
    \ClassifyingA
  })
  :=
  \pi_0
  \,\mathrm{Map}\bracket({
    X, \ClassifyingA
  })
  \mathrlap{\,.}
\end{equation}

\paragraph
{Quaternionic Projective Spaces}
We recall that $S^4$ is homeomorphically the quaternionic projective line,  
while $B \mathrm{SU}(2)$ is the infinite projective space (cf. \cite[Rem. 3.88]{SS23-Mf}), and that under this identification the canonical inclusion map represents the generator $1 \in \mathbb{Z} \simeq \pi_4\bracket(B{ \mathrm{SU}(2) })$:
\begin{equation}
\label{UnitMapFromS4ToBSU2}
\begin{tikzcd}[
  row sep=6pt
]
  S^4
  \ar[
    rrr,
    "{ 1 }"
  ]
  \ar[
    d,
    phantom,
    "{ \simeq }"{rotate=-90}
  ]
  &&&
  B \mathrm{SU}(2)
  \ar[
    d,
    phantom,
    "{ \simeq }"{rotate=-90}
  ]
  \\
  \mathbb{H}P^1
  \ar[
    r,
    hook,
    "{ \iota_1 }"
  ]
  &
  \mathbb{H}P^2
  \ar[
    r, 
    hook,
    "{ \iota_2 }"
  ]
  &
  \cdots
  \ar[r, hook]
  &
  \mathbb{H}P^\infty
  \mathrlap{\,.}
\end{tikzcd}
\end{equation}
Concretely, $\mathbb{H}P^2$ is obtained from $\mathbb{H}P^1$ by attaching an 8-cell along the quaternionic Hopf fibration $h_{\mathbb{H}}$ \cref{TheQuaternionicHopfFibration}:
\begin{equation}
\label
{CellAttachmentForHP2}
  \begin{tikzcd}[row sep=15pt, column sep=large]
    S^7
    \ar[
      r,
      "{ h_{\mathbb{H}} }"
    ]
    \ar[
      d,
      hook
    ]
    \ar[
      dr,
      phantom,
      "{ \ulcorner }"{pos=.9}
    ]
    & 
    \mathbb{H}P^1
    \ar[
      d, 
      hook,
      "{ \iota_1 }"
    ]
    \\
    D^8
    \ar[
      r,
    ]
    & 
    \mathbb{H}P^2
    \mathrlap{\,.}
  \end{tikzcd}
\end{equation}
Next, $\mathbb{H}P^3$ is obtained by adding a further 12-cell and so on.

In view of this cell decomposition, one finds that:
\begin{enumerate}
\item 
the map $\inlinetikzcd{\mathbb{HP}^2 \ar[r, hook] \&  \mathbb{H}P^\infty }$ induces isomorphisms on homotopy up to dimension 10 (by the cellular approximation theorem, cf. \cite[Thm. 4.8]{Hatcher2002})
\begin{equation}
\label
{HomotopyOfHP2InsideHPInfty}
  \begin{tikzcd}
    \pi_{\leq 10}
    \, 
    \mathbb{H}P^2
    \ar[
      r,
      "{ \sim }"
    ]
    &
    \pi_{\leq 10}
    \,
    \mathbb{H}P^\infty
    \mathrlap{\,,}
  \end{tikzcd}
\end{equation}
\item 
the integral cohomology ring of quaternionic complex projective spaces is (as in \cite[Thm. 3.19]{Hatcher2002}) freely generated by a degree=4 generator, which under the above equivalence we may address as the second Chern class $c_2$ (as in  \cite[\S 14]{Milnor1974}):
\begin{equation}
\label
{IntegralCohomologyOfBSU2}
  c_2 
  \in
  H^4\bracket({
    B \mathrm{SU}(2);
    \mathbb{Z}
  })
  \subset
  H^\bullet\bracket({
    B \mathrm{SU}(2);
    \mathbb{Z}
  })  
  \simeq
  \mathbb{Z}[c_2]
  \mathrlap{\,.}
\end{equation}
(alternatively, under the further homeomorphism $\mathrm{SU}(2) \simeq \mathrm{Spin}(3)$, as the first fractional Pontrjagin class $\tfrac{1}{2}p_1$). 

Using the same symbol also for the pullback of this generator to any finite stage, we have the cohomology rings
\begin{equation}
  H^\bullet\bracket({
    \mathbb{H}P^n
    ;
    \mathbb{Z}
  })
  \simeq
  \mathbb{Z}[{
    c_2
  }]\big/\bracket({
    (c_2)^{n+1}
  })
  \mathrlap{\,.}
\end{equation}
In particular (and trivially so):
\begin{equation}
\label
{IntegralCohomologyRingOfS4}
  H^4\bracket({
    S^4;
    \mathbb{Z}
  })
  \simeq
  \mathbb{Z}[c_2]/\bracket({
    (c_2)^2
  })
  \mathrlap{\,.}
\end{equation}
\end{enumerate}

\paragraph
{The 6-Group Classifying Spaces}
By \cref{OrdinaryCohomologyClassified} the squared universal second Chern class has a classifying map, which we denote by the same symbol:
\begin{equation}
\label
{ClassifyingMapForC2Square}
  \begin{tikzcd}
    B \mathrm{SU}(2)
    \ar[r]
    \ar[
      rrr,
      downhorup,
      "{ 
        (c_2)^2
      }"{description}
    ]
    &
    B \mathrm{U}(2)
    \ar[
      rr,
      "{
        (c_2)^2
      }"
    ]
    &&
    B^7
    \mathbb{Z}
    \mathrlap{\,.}
  \end{tikzcd}
\end{equation}
The homotopy fibers of these maps are our hatted spaces $B \widehat{\mathrm{U}(2)}$ \cref{TheSpaceBHatU2} and $B \widehat{\mathrm{SU}(2)}$ \cref{TheSpaceHatBSU2}, which thereby come with a canonical comparison map:
\begin{lemma}
  This canonical map 
  \begin{equation}
    \label{MapFromBHatSU2ToBhatU2}
    \begin{tikzcd}
     B \widehat{\mathrm{SU}(2)}
     \ar[r]
     &
     B \widehat{\mathrm{U}(2)}
  \end{tikzcd}
  \end{equation}
  is a homotopy $\mathrm{U}(1)$-fibration.
\end{lemma}
\begin{proof}
First, we note that the unhatted map is already a homotopy $\mathrm{U}(1)$-fibration (cf. \cite[Lem. 2.7(b)]{FSS20-H}):
\begin{equation}
\label
{CircleAsHomotopyFiberOfBSU2ToBU2}
  \mathrm{fib}\big(
    \inlinetikzcd{
      B \mathrm{SU}(2)
      \ar[r]
      \&
      B \mathrm{U}(2)
     }
   \big)
   \simeq
   \mathrm{U}(2)/\mathrm{SU}(2)
   \simeq
   \mathrm{U}(1)
   \mathrlap{\,.}
\end{equation}
Now, the fiber in question may be understood as the homotopy fiber $X$ in the left vertical sequence of the following homotopy-commuting diagram
\begin{equation}
  \begin{tikzcd}[row sep=small]
    X
    \ar[r]
    \ar[d]
    &
    \mathrm{U}(1)
    \ar[d]
    \ar[r]
    &
    \ast
    \ar[d]
    \\
    B \widehat{\mathrm{SU}(2)}
    \ar[r]
    \ar[d]
    &
    B \mathrm{SU}(2)
    \ar[
      r,
      "{
        (c_2)^2
      }"
    ]
    \ar[d]
    &
    B^8 \mathbb{Z}
    \ar[
      d,
      equals
    ]
    \\
    B \widehat{\mathrm{U}(2)}
    \ar[r]
    &
    B \mathrm{U}(2)
    \ar[
      r,
      "{
        (c_2)^2
      }"
    ]
    &
    B^8 \mathbb{Z}
    \mathrlap{\,.}
  \end{tikzcd}
\end{equation}
Here the middle and bottom rows are homotopy fiber sequences by definition, and the middle and right columns are homotopy fiber sequences by \cref{CircleAsHomotopyFiberOfBSU2ToBU2}. But since homotopy limits commute with each other, it follows that:
  \begin{align*}
    X
    & 
    \defneq
    \mathrm{fib}\Big(
      \inlinetikzcd{
        B \widehat{\mathrm{SU}(2)}
        \ar[r]
        \&
        B \widehat{\mathrm{U}(2)}
      }
    \Big)
    \\
    & \simeq
    \mathrm{fib}\big(
      \inlinetikzcd{
        \mathrm{U}(1)
        \ar[r]
        \&
        \ast
      }
    \big)
    \\
    & 
    \simeq
    \mathrm{U}(1)
    \mathrlap{\,.}
    \qedhere 
  \end{align*}
\end{proof}
\begin{corollary}
\label[corollary]
{TheExtCycAdjunct}
  Under the $\mathrm{Ext}/\mathrm{Cyc}$-adjunction \textup{(\parencites[\S 2.2]{BMSS2019}[\S 2.2]{SS24-Cyc})}, the adjunct of the identity on $B \widehat{\mathrm{SU}(2)}$ is a map of the form \cref{TheAdjunctMapIntoCycBHatSU2}:
  \begin{equation}
    \label{TheAdjunctMapIntoCycBHatSU2Derived}
    \begin{tikzcd}[
      column sep=15pt,
      row sep=0pt
    ]
      B \widehat{\mathrm{U}(2)}
      \ar[rr]
      \ar[
        dr
      ]
      &&
      \mathrm{Cyc}\bracket({
        B \widehat{\mathrm{SU}(2)}
      })  \mathrlap{\,.}
      \ar[
        dl
      ]
      \\
      & 
      B \mathrm{U}(1)
    \end{tikzcd}
  \end{equation}
\end{corollary}

\paragraph
{Blakers-Massey theorem}
Recall that:
\begin{enumerate}
\item a topological space is called:
\begin{itemize}
  \item  \emph{$(-2)$-connected}: always,
  \item \emph{$(-1)$-connected} if it is not empty,
  \item  \emph{$0$-connected} if it is $(-1)$-connected and path connected,
  \item  \emph{$(k \geq 1)$-connected} if it is $0$-connected and all its $\pi_{\leq k}$ are trivial;
\end{itemize}
\item a map between spaces is called: 
\begin{itemize}
  \item \emph{$k$-connected} 
  if all its homotopy fibers are $(k-1)$-connected.
\end{itemize}
\end{enumerate}
So by the homotopy LES, a $k$-connected map induces isomorphisms on $\pi_{\leq k-1}$.

\begin{proposition}[Blakers-Massey theorem {\parencites[p. 309]{Goodwillie1991CalculusII}{AnelBiedermanFinsterJoyal2020}}]
\label[proposition]
{BlakersMasseyTheorem}
Given a homotopy pushout square, where $f_i$ is $k_i$-connected,
\begin{equation}
  \begin{tikzcd}[row sep=2pt, column sep=12pt]
    X
    \ar[rr, "{ f_2 }"]
    \ar[dd, "{ f_1 }"{swap} ]
    \ar[
      dr,
      dashed,
      shorten >=-3pt
    ]
    &&
    B
    \ar[dd]
    \\
    & 
    A 
    \underset
      {\smash{
        A \underset{X}{\sqcup} B
      }}
      {\times}
    B
    \ar[ur]
    \ar[
      dl,
      shorten <=-9pt
    ]
    \\
    A
    \ar[rr]
    &&
    A \underset{X}{\sqcup} B
    \mathrlap{\,,}
  \end{tikzcd}
\end{equation}
the dashed comparison map to the homotopy pullback is $(k_1 + k_2 - 1)$-connected.
\end{proposition}

\begin{example}
 Applied to the cell attachment \cref{CellAttachmentForHP2}, which is evidently the homotopy pushout of a 3-connected map along a 7-connected map, 
 the BM theorem (\cref{BlakersMasseyTheorem}) gives that the comparison map to the homotopy fiber $\mathrm{fib}(\iota_1)$ of the first cell inclusion:
 \begin{equation}
   \begin{tikzcd}[row sep=2pt, column sep=12pt]
     S^7
     \ar[
       rr,
       "{ h_{\mathbb{H}} }"
     ]
     \ar[dd]
     \ar[
       dr,
       dashed,
       shorten >=-2pt
     ]
     &&
     S^4
     \ar[
       dd,
       "{ \iota_1 }"
     ]
     \\
     &
     \mathrm{fib}(1)
     \ar[
       ur,
       shorten <=-2pt
      ]
     \ar[
       dl,
       shorten =-4pt
     ]
     \\
     \ast
     \ar[rr]
     &&
     \mathbb{H}P^2
   \end{tikzcd}
 \end{equation}
 is 9-connected, hence induces isomorphisms
 \begin{equation}
   \label
   {LowHomotopyIsosBetweenS7AndFiber}
   \begin{tikzcd}
     \pi_{\leq 8}\,
     S^7
     \ar[
       r,
       "{ \sim }"
      ]
     &
     \pi_{\leq 8}\,
     \mathrm{fib}(\iota_1)
     \mathrlap{\,.}
   \end{tikzcd}
 \end{equation}
\end{example}

\begin{lemma}
  The homotopy fibers of $\inlinetikzcd{S^4 \ar[r, "{ \iota_1 }"] \& \mathbb{H}P^2}$ and of $\inlinetikzcd{S^4 \ar[r, "{ 1 }"] \& \mathbb{H}P^\infty}$ are identified in low degrees:
  \begin{equation}
  \label
  {LowHomotopyIsosBetweenFibers}
    \begin{tikzcd}
      \pi_{\leq 9}\,
      \mathrm{fib}(\iota_1)
      \ar[r, "{ \sim }"]
      &
      \pi_{\leq 9}\,
      \mathrm{fib}(1)
      \mathrlap{\,.}
    \end{tikzcd}
  \end{equation}
\end{lemma}
\begin{proof}
The map of homotopy fiber sequences
\begin{equation}
  \begin{tikzcd}[
    row sep=7pt
  ]
    \mathrm{fib}(\iota_1)
    \ar[r]
    \ar[d]
    &
    S^4 
    \ar[
      r,
      "{ \iota_1 }"
    ]
    \ar[
      d,
      equals
    ]
    &
    \mathbb{H}P^2
    \ar[
      d,
      hook
    ]
    \\
    \mathrm{fib}(1)
    \ar[r]
    &
    S^4
    \ar[
      r,
      "{ 1 }"
    ]
    & 
    \mathbb{H}P^\infty
  \end{tikzcd}
\end{equation}
induces the following map of corresponding homotopy LESs:
\begin{equation}
  \begin{tikzcd}[
    row sep=7pt
  ]
    \pi_{k+1}\,
    S^4
    \ar[r]
    \ar[
      d,
      equals
    ]
    &
    \pi_{k+1}\,
    \mathbb{H}P^2
    \ar[
      r
    ]
    \ar[d]
    &
    \pi_k\,
    \mathrm{fib}(\iota_1)
    \ar[r]
    \ar[d]
    &
    \pi_k\, S^4 
    \ar[
      r,
      "{ (\iota_1)\ast }"
    ]
    \ar[
      d,
      equals
    ]
    &
    \pi_k\,
    \mathbb{H}P^2
    \ar[
      d
    ]
    \\
    \pi_{k+1}\, 
    S^4
    \ar[r]
    &
    \pi_{k+1}\,
    \mathbb{H}P^\infty
    \ar[r]
    &
    \pi_k\,
    \mathrm{fib}(1)
    \ar[r]
    &
    \pi_k\,
    S^4
    \ar[
      r,
      "{ 1_\ast }"
    ]
    & 
    \pi_k\,
    \mathbb{H}P^\infty
    \mathrlap{\,.}
  \end{tikzcd}
\end{equation}
Hence the five-lemma (cf. \cite[Ex. 1.3.3]{Weibel1994}) implies the claim, by \cref{HomotopyOfHP2InsideHPInfty}.
\end{proof}

\paragraph
{The Comparison Map}
By  \cref{IntegralCohomologyRingOfS4,OrdinaryCohomologyClassified},
the canonical map \cref{UnitMapFromS4ToBSU2} has null-homotopic postcomposition with $(c_2)^2$ and hence factors, up to homotopy, through the latter's homotopy fiber \cref{TheSpaceHatBSU2} via an essentially unique map $\hat{1}$:
\begin{equation}
\label
{ComparisonMapFromUniversalProperty}
  \begin{tikzcd}[
    row sep=20pt 
  ]
    & 
    B \widehat{\mathrm{SU}(2)}
    \ar[
      d,
      "{ 
        \BHatSUTwoFibration 
      }"{swap,
        pos=.4
      }
    ]
    \ar[r]
    \ar[
      dr,
      phantom,
      "{ \lrcorner_{h} }"{pos=.1}
    ]
    &
    \ast
    \ar[d]
    \\
    S^4
    \ar[
      r,
      "{ 1 }"
    ]
    \ar[
      ur,
      dashed,
      "{ \hat 1 }"
    ]
    & 
    B
    \mathrm{SU}(2)
    \ar[
      r,
      "{
        (c_2)^2
      }"
    ]
    &
    B^8 \mathbb{Z}
    \mathrlap{\,.}
  \end{tikzcd}
\end{equation}

\begin{theorem}
\label[theorem]
{ComparisonMapIs7Equivalence}
  The comparison map $\hat 1$ \cref{ComparisonMapFromUniversalProperty} induces isomorphisms on $\pi_{\leq 7}$.
\end{theorem}
\begin{proof}
  In degrees $n \leq 6$ this is straightforward: The map $1$ induces isomorphisms in these degrees (since the next cell of $B \mathrm{SU}(2)$ is in dimension 8) as does the map $\mathrm{fib}$ (by the long exact sequence of homotopy groups and using that $\pi_{\leq 7}\bracket({B^8 \mathbb{Z}})$ is trivial).

  More work is required to show isomorphy on $\pi_7$.
To this end, consider the following map of homotopy fiber sequences induced from the diagram \cref{ComparisonMapFromUniversalProperty}:
\begin{equation}
\label
{MapOfHomotopyFiberSequences}
  \begin{tikzcd}[
    row sep=12pt
  ]
    \mathrm{fib}(1)
    \ar[r]
    \ar[
      d,
      "{ \phi }"
    ]
    &
    S^4 
    \ar[
      r,
      "{ 1 }"
    ]
    \ar[
      d,
      "{ \hat 1 }"
    ]
    &
    B \mathrm{SU}(2)
    \ar[
      d,
      equals
    ]
    \\
    B^7 \mathbb{Z}
    \ar[r]
    &
    B \widehat{\mathrm{SU}(2)}
    \ar[
      r,
      "{ 
        \BHatSUTwoFibration 
      }"
    ]
    &
    B \mathrm{SU}(2)
  \end{tikzcd}
\end{equation}
The corresponding morphism of induced homotopy LESs is of this form:
\begin{equation}
\label
{MapOfLESToIsolatePi7}
  \begin{tikzcd}[row sep=12pt]
    \pi_8\, 
    B \mathrm{SU}(2)
    \ar[r]
    \ar[
      d,
      equals
    ]
    &
    \overbrace{
    \pi_7\,
    \mathrm{fib}(1)
    }^{ \mathbb{Z} }
    \ar[r]
    \ar[
      d,
      "{ \phi_\ast }"
    ]
    &
    \pi_7\,
    S^4 
    \ar[
      r,
      "{ 1_\ast }"
    ]
    \ar[
      d,
      "{ 
        \hat 1_\ast 
      }"
    ]
    &
    \pi_7\,
    B \mathrm{SU}(2)
    \ar[
      d,
      equals
    ]
    \ar[
      r
    ]
    &
    \overbrace{
    \pi_6\,
    \mathrm{fib}(1)
    }^{ 0 }
    \ar[d]
    \\
    \pi_8\,
    B \mathrm{SU}(2)
    \ar[r]
    &
    \underbrace{
    \pi_7\,
    B^7 \mathbb{Z}
    }_{ \mathbb{Z} }
    \ar[r]
    &
    \pi_7\,
    B \widehat{\mathrm{SU}(2)}
    \ar[
      r,
      "{ 
        \BHatSUTwoFibration_\ast
      }"
    ]
    &
    \pi_7\,
    B \mathrm{SU}(2)
    \ar[r]
    &
    \underbrace{
    \pi_6\,
    B^7 \mathbb{Z}
    }_{ 0 }
    \mathrlap{\,,}
  \end{tikzcd}
\end{equation}
where over the braces we used \cref{LowHomotopyIsosBetweenFibers,LowHomotopyIsosBetweenS7AndFiber}.

Hence the five-lemma (cf. \cite[Ex. 1.3.3]{Weibel1994}) implies the claim as soon as we show that $\phi_\ast$ here is an isomorphism. We relegate this to \cref{PhiHatIsIsoInDegree7}.
\end{proof}

\begin{lemma}
\label[lemma]
{PhiHatIsIsoInDegree7}
  The map $\phi_\ast$ in \cref{MapOfLESToIsolatePi7} is an isomorphism.
\end{lemma}
\begin{proof}
Consider now the map of cohomology Serre LESs (cf. \cite[Ex. 5.D with Thm. 5.2]{McCleary2001}) induced by \cref{MapOfHomotopyFiberSequences}:
\begin{equation}
\label
{MapOfSerreLESs}
  \begin{tikzcd}[row sep=12pt]
    \overbrace{
    H^8\bracket({
      S^4;
      \mathbb{Z}
    })
    }^{ 0 }
    &
    H^8\bracket({
      B \mathrm{SU}(2);
      \mathbb{Z}
    })
    \ar[
      l,
      "{ 1^\ast }"{swap}
    ]
    &
    H^7\bracket({
      \mathrm{fib}(1);
      \mathbb{Z}
    })
    \ar[
      l,
      "{ \tau_8 }"{swap},
      "{ \sim }"
    ]
    &
    \overbrace{
    H^7\bracket({
      S^4;
      \mathbb{Z}
    })}^{0}
    \ar[
      l
    ]
    \\
    & 
    H^8\bracket({
      B\mathrm{SU}(2);
      \mathbb{Z}
    })
    \ar[
      u,
      equals
    ]
    &    
    H^7\bracket({
      B^7 \mathbb{Z};
      \mathbb{Z}
    })
    \mathrlap{\,.}
    \ar[
      l,
      "{
        \tau'_8
      }"{swap}
    ]
    \ar[
      u,
      "{ 
        \phi^\ast 
      }"{swap, pos=.4}
    ]
  \end{tikzcd}
\end{equation}
Here the top transgression map, $\tau_8$, is immediately seen to be an isomorphism by exactness, and the bottom map may be checked to be an isomorphism, too (\cref{TauPrime8IsIso}). By commutativity, this implies that $\phi^\ast$ is an isomorphism. From this, the Hurewicz theorem (cf. \cite[Thm. 4.32]{Hatcher2002}) implies that also $\phi_\ast$ is an isomorphism, since $B^7 \mathbb{Z}$ and $\mathrm{fib}(1)$ are 6-connected by \cref{HomotopyGroupsOfEMSpace,LowHomotopyIsosBetweenS7AndFiber}.
\end{proof}
\begin{lemma}
\label[lemma]
{TauPrime8IsIso}
  The map $\tau'_8$ in \cref{MapOfSerreLESs} is an isomorphism.
\end{lemma}
\begin{proof}
By construction \cref{ComparisonMapFromUniversalProperty}, the bottom homotopy fiber sequence in \cref{MapOfHomotopyFiberSequences} is equivalently the pullback of the path fibration over $B^8 \mathbb{Z}$:
\begin{equation}
  \begin{tikzcd}[
    row sep=12pt
  ]
    B^7 \mathbb{Z}
    \ar[
      r,
      equals
    ]
    \ar[d]
    &
    B^7 \mathbb{Z}
    \ar[d]
    \\
    B \widehat{\mathrm{SU}(2)}
    \ar[
      d,
      "{ 
        \BHatSUTwoFibration 
      }"{pos=.4}
    ]
    \ar[r]
    \ar[
      dr,
      phantom,
      "{ \lrcorner }"{pos=.1}
    ]
    &
    P B^8 \mathbb{Z}
    \ar[
      d
    ]
    \\
    B\mathrm{SU}(2)
    \ar[
      r,
      "{
        (c_2)^2
      }"
    ]
    &
    B^8 \mathbb{Z}
    \mathrlap{\,.}
  \end{tikzcd}
\end{equation}
Hence, again by the naturality of the cohomology Serre LES (cf. \cite[Ex. 5.D with Thm. 5.2]{McCleary2001}), we have the commuting square
\begin{equation}
  \begin{tikzcd}
    H^8\bracket({
      B \mathrm{SU}(2);
      \mathbb{Z}
    })
    &
    H^7\bracket({
      B^7\mathbb{Z};
      \mathbb{Z}
    })
    \ar[
      l,
      "{ \tau'_8 }"{swap}
    ]
    \\
    H^8\bracket({
      B^8 \mathbb{Z};
      \mathbb{Z}
    })
    \ar[
      u,
      "{
        \left((c_2)^2\right)^\ast
      }"{swap},
      "{ \sim }"{sloped}
    ]
    &
    H^7\bracket({
      B^7\mathbb{Z};
      \mathbb{Z}
    })
    \mathrlap{\,.}
    \ar[
      u,
      equals,
    ]
    \ar[
      l,
      "{
        \tau^{\mathrm{univ}}_8
      }"{swap},
      "{ \sim }"
    ]    
  \end{tikzcd}
\end{equation}
But as indicated, here the bottom map is an isomorphism (by \cite[Lem. 65.2]{Miller2021}) as is the left map, by \cref{IntegralCohomologyOfBSU2}. This implies the claim, by commutativity.
\end{proof}

This completes the proof of \cref{ComparisonMapIs7Equivalence}, establishing the integral 7-equivalence. Next we extend this to a 7-equivalence of classifying fibrations.

\paragraph
{Relation to the Quaternionic Hopf fibration}

We note that the homotopy fiber of our unit map \cref{UnitMapFromS4ToBSU2} is the unit class $\mathrm{SU}(2)$-principal bundle over $S^4$, which is the quaternionic Hopf fibration $h_{\mathbb{H}}$ (cf. \parencites[(9.62)]{Nakahara2018}[\S 5.2.3]{SS26-Orb}):
\begin{equation}
\label
{TheQuaternionicHopfFibration}
  \begin{tikzcd}[
    row sep=12pt
  ]
    S^7
    \ar[r]
    \ar[
      d,
      "{ 
        h_{\mathbb{H}} 
      }"{swap}
    ]
    \ar[
      dr,
      phantom,
      "{ \lrcorner_h }"{pos=.1}
    ]
    &
    \ast
    \ar[d]
    \\
    S^4
    \ar[
      r,
      "{ 1 }"
    ]
    &
    B \mathrm{SU}(2)
  \end{tikzcd}
\end{equation}
Its hatted analogue in our context is hence the homotopy fiber $\HatESUTwoFibration$ of $\BHatSUTwoFibration$ \cref{ComparisonMapFromUniversalProperty}:
\begin{equation}
\label
{FibrationOfHattedSpaces}
  \begin{tikzcd}[
    row sep=12pt
  ]
    B^7 \mathbb{Z}
    \ar[r]
    \ar[
      d,
      "{
        \HatESUTwoFibration
      }"{swap}
    ]
    \ar[
      dr,
      phantom,
      "{ \lrcorner_h }"{pos=.1}
    ]
    &
    \ast
    \ar[d]
    \\
    B \widehat{\mathrm{SU}(2)}
    \ar[
      r,
      "{ 
        \BHatSUTwoFibration 
      }"
    ]
    &
    B \mathrm{SU}(2)
    \mathrlap{\,.}
  \end{tikzcd}
\end{equation}
By the homotopy commutativity of \cref{ComparisonMapFromUniversalProperty} and the universal property of homotopy fibers, we have an essentially unique canonical comparison map $\hat h_{\mathbb{H}}$ between these two fibrations:
\begin{equation}
\label
{ComparisonMapFromS7}
  \begin{tikzcd}[
    row sep=12pt
  ]
    S^7
    \ar[
      r,
      "{ h_{\mathbb{H}} }"
    ]
    \ar[
      d,
      dashed,
      "{
        \hat h_{\mathbb{H}}
      }"
    ]
    & 
    S^4
    \ar[
      r,
      "{ 1 }"
    ]
    \ar[
      d,
      "{ \hat 1 }"
    ]
    &
    B \mathrm{SU}(2)
    \ar[
      d,
      equals
    ]
    \\
    B^7 \mathbb{Z}
    \ar[
      r,
      "{
        \HatESUTwoFibration
      }"
    ]
    &
    B \widehat{\mathrm{SU}(2)}
    \ar[
      r,
      "{ \BHatSUTwoFibration }"
    ]
    &
    B \mathrm{SU}(2)
    \mathrlap{\,.}
  \end{tikzcd}
\end{equation}

\begin{theorem}
\label[theorem]
{ComparisonOfFibrationsIs7Equivalence}
  The comparison map $\hat h_{\mathbb{H}}$ \cref{ComparisonMapFromS7}
  induces isomorphisms on $\pi_{\leq 7}$.
\end{theorem}
\begin{proof}
The map of homotopy LESs induced by \cref{ComparisonMapFromS7} is of this form:
\begin{equation}
  \begin{tikzcd}[
    column sep=20pt
  ]
    \pi_{k+1}\,
    S^4
    \ar[
      r,
      "{ 1_\ast }"
    ]
    \ar[
      d,
      "{ \hat 1_\ast }"
    ]
    &
    \pi_{k+1}\,
    B \mathrm{SU}(2)
    \ar[
      r,
      "{ \partial_k }"
    ]
    \ar[
      d,
      equals
    ]
    &
    \pi_k\,
    S^7
    \ar[
      r,
      "{ 
        (h_{\mathbb{H}})_\ast 
      }"
    ]
    \ar[
      d,
      dashed,
      "{
        (\hat h_{\mathbb{H}})_\ast
      }"
    ]
    & 
    \pi_k\,
    S^4
    \ar[
      r,
      "{ 1_\ast }"
    ]
    \ar[
      d,
      "{ \hat 1_\ast }"
    ]
    &
    \pi_k\,
    B \mathrm{SU}(2)
    \ar[
      d,
      equals
    ]
    \\
    \pi_{k+1}\,
    B \widehat{\mathrm{SU}(2)}
    \ar[
      r,
      "{
        \BHatSUTwoFibration
      }"
    ]
    &
    \pi_{k+1}\,
    B \mathrm{SU}(2)
    \ar[
      r,
      "{
        \partial'_k
      }"
    ]
    &
    \pi_k\,
    B^7 \mathbb{Z}
    \ar[
      r,
      "{
        \HatESUTwoFibration_\ast
      }"
    ]
    &
    \pi_k\,
    B \widehat{\mathrm{SU}(2)}
    \ar[
      r,
      "{ 
        \BHatSUTwoFibration_\ast 
      }"
    ]
    &
    \pi_k\,
    B \mathrm{SU}(2)
    \mathrlap{\,.}
  \end{tikzcd}
\end{equation}
Hence by the five-lemma (cf. \cite[Ex. 1.3.3]{Weibel1994}) it is now sufficient to show that $\hat 1_\ast$ is an isomorphism on $\pi_{\leq 7}$ and an epimorphism on $\pi_8$. The first statement is \cref{ComparisonMapIs7Equivalence}, the second we relegate to \cref{HatOneInducesEpimorphismOnPi8}.
\end{proof}
\begin{lemma}
\label[lemma]
{HatOneInducesEpimorphismOnPi8}
  The map $\hat 1$ \cref{ComparisonMapFromUniversalProperty} induces an epimorphism on $\pi_8$. 
\end{lemma}
\begin{proof}
The map of homotopy LESs induced by \cref{ComparisonMapFromS7} looks in the relevant parts as follows:
\begin{equation}
  \begin{tikzcd}[
    row sep=12pt
  ]
    &
    \pi_8\,
    S^4
    \ar[
      r,
      ->>,
      "{
        1_\ast
      }"
    ]
    \ar[
      d,
      "{
        \hat 1_\ast
      }"
    ]
    &
    \overbrace{
      \pi_8\,
      B\mathrm{SU}(2)
    }^{
      \mathrm{finite}
    }
    \ar[
      r,
      "{ \partial_7 }",
      "{ 0 }"{swap}
    ]
    \ar[
      d,
      equals
    ]
    &
    \overbrace{
      \pi_7\,
      S^7
    }^{ \mathbb{Z} }
    \\
    \underbrace{
    \pi_8\,
    B^7 \mathbb{Z}
    }_{ 0 }
    \ar[r]
    &
    \pi_8\,
    B \widehat{\mathrm{SU}(2)}
    \ar[
      r,
      "{
        \BHatSUTwoFibration_\ast
      }",
      "{ \sim }"{swap}
    ]
    &
    \underbrace{
    \pi_8\,
    B \mathrm{SU}(2)
    }_{ \mathrm{finite} }
    \ar[
      r,
      "{ \partial'_7 }",
      "{ 0 }"{swap}
    ]
    &
    \underbrace{
      \pi_7\,
      B^7 \mathbb{Z}
    }_{\mathbb{Z}}
    \mathrlap{\,.}
  \end{tikzcd}
\end{equation}

Using Serre finiteness \cref{SerreFiniteness} on the terms with braces implies that the shown connecting homomorphisms both vanish, whence exactness implies that $1_\ast$ is an epimorphism and $\BHatSUTwoFibration_\ast$ is an isomorphism. But then commutativity implies the claim.
\end{proof}

This establishes that we have a 7-equivalence of classifying fibrations.
Next we turn to the rationalization of this situation.

\paragraph
{Rational Homotopy}
For general background on rational homotopy theory, cf. \parencites{FHT2000}{Hess2007}[\S 5]{FSS23-Char} and for exposition in our context see \cite{FSS19-RationalM}. 
Here we denote the \emph{minimal Sullivan model} (cf. \cite{Menichi2015}) for the $\mathbb{R}$-rational homotopy type of a space $\ClassifyingB$ 
(simply connected and of finite rational type)
by: 
\begin{equation}
\mathrm{CE}\bracket({
  \mathfrak{l}
  \ClassifyingB
})
  \in 
  \mathrm{dgcAlg}_{_{\mathbb{R}}}
  \,,
\end{equation}
since it is the \emph{Chevalley-Eilenberg algebra} of the \emph{real Whitehead-bracket $L_\infty$-algebra} $\mathfrak{l}\ClassifyingB$ \cite[Prop. 5.11, Rem. 5.4]{FSS23-Char}. Moreover, when we make these CE-algebras explicit we write (as in \cite{GSS25-TD}): 
\begin{equation}
\label
{PresentationOfCEAlgebra}
  \mathrm{CE}\bracket({
    \mathfrak{l}
    \ClassifyingB
  })
  \defneq
  \mathbb{R}_{_{\mathrm{d}}}
  \big[
    \text{generators}
  \big]\big/
  \big(
    \text{differential relations}
  \big)
  \mathrlap{\,,}
\end{equation}
indicating how they are \emph{free differential algebras} (FDA in supergravity terminology, cf. \cite{CastellaniDAuria2025,FSS19-RationalM,FSS15-WZW}) on the given graded generators, quotiented by the given differential ideal. 

The subscript on a generator will always indicate its degree, directly analogous to the notation for the flux densities in \cref{OnTheChargeQuantization}. Indeed, $\ClassifyingB$ is the classifying space of an admissible flux-quantization law if 
\parencites[\S 3.1-2]{SS25-Flux}{SS24-Phase}
there is a presentation \cref{PresentationOfCEAlgebra} such that these generators and relations correspond to the given flux species and their Bianchi/Gauss identities, respectively:
\begin{equation}
\label
{GaussLawFromRelations}
  \begin{tikzcd}[sep=0pt]
    \text{$\mathrm{CE}\bracket({\mathfrak{l}\ClassifyingB})$-generators}
    &\leftrightarrow&
    \text{flux species}
    \\
    \text{$\mathrm{CE}\bracket({\mathfrak{l}\ClassifyingB})$-relations}
    &\leftrightarrow&
    \substack{
    \text{duality-symmetric Bianchi identities}
    \\
    \text{/ Gauss laws on Cauchy surface.}
    }
  \end{tikzcd}
\end{equation}

More generally, for $\inlinetikzcd{ \ClassifyingA \ar[r, "{ \wp }"] \& \ClassifyingB }$ a fibration of spaces, we denote the \emph{relative} minimal Sullivan model of $\ClassifyingA$, $\wp$-relative to the minimal Sullivan model of $\ClassifyingB$, by (cf. \cite[Prop. 5.16]{FSS23-Char}):
\begin{equation}
\label
{RelativeMinimalSullivanModel}
  \mathrm{CE}\bracket({
    \mathfrak{l}_{_{\ClassifyingB}}
    \ClassifyingA
  })
  \in
  \mathrm{dgcAlg}_{_{\mathbb{R}}}
  \mathrlap{\,,}
\end{equation}
obtained by adjoining a minimal set of \emph{further} generators to the minimal model \cref{PresentationOfCEAlgebra} of the base (without removing any of these generators).

Hence, while contractible pairs of generators, of the form
$\big(\mathrm{d}\, e_p = f_{p+1}
\,,
\mathrm{d}\, f_{p+1}  = 0\big)$, do not appear in minimal Sullivan models (being non-minimal), they may appear in relative minimal Sullivan models, namely whenever the cohomology class of $f_{p+1}$ in the base space vanishes when pulled back along $\wp$.

Now, first to recall that the classical \emph{Serre finiteness theorem} (cf. \cite[\S I, Thm. 1.1.8]{Ravenel2004}) says that the homotopy groups of spheres are mostly all finite. Concretely, for $n \in \mathbb{N}$:
\begin{subequations}
\label
{SerreFiniteness}
\begin{align}
  \pi_k\,
  S^{2n+1}
  &
  \simeq
  \begin{cases}
    \mathbb{Z} & 
    \text{if } k = 2n+1
    \\[-2pt]
    \text{finite} & \text{otherwise,}
  \end{cases}
  \\
  \pi_k\,
  S^{2(n+1)}
  &
  \simeq
  \begin{cases}
    \mathbb{Z} 
    & 
    \text{if } k = 2(n+1)
    \\[-2pt]
    \mathbb{Z} 
    \times
    \mathrm{finite}
    & 
    \text{if } k = 4(n+1)-1
    \\[-2pt]
    \text{finite} & \text{otherwise.}
  \end{cases}
\end{align}
\end{subequations}
Accordingly, all finite group factors (\emph{torsion subgroups}) disappear under rationalization, and the Sullivan minimal models of spheres are (cf. \cite[\S 1.2]{Menichi2015}) concentrated on just the one or two remaining free homotopy groups:
\begin{subequations}
\label
{MinimalSullivanModelsOfSpheres}
\begin{align}
  \mathrm{CE}\bracket({
    \mathfrak{l}S^{2n+1}
  })
  &
  \simeq
  \mathbb{R}_{_{\mathrm{d}}}
  [
  g_{2n+1}
  ]\big/
  \bracket({
    \mathrm{d}\, g_{2n+1}
    = 
    0
  })
  \\
  \mathrm{CE}\bracket({
    \mathfrak{l}S^{2(n+1)}
  })
  &
  \simeq
  \mathbb{R}_{_{\mathrm{d}}}
  \left[
    \begin{aligned}
      &g_{2(n+1)}
      \\[-2pt]
      &g_{4(n+1)-1}
    \end{aligned}
  \right]\big/
  \left({
    \begin{aligned}
    \mathrm{d}\, g_{2(n+1)}
    &= 
    0
    \\[-2pt]
    \mathrm{d}\, g_{4(n+1)-1}
    & =
    \tfrac{1}{2}
    g_{2(n+1)}^2
    \end{aligned}
  }\right)
  \mathrlap{\,.}
\end{align}
\end{subequations}

With this, \cref{ComparisonMapIs7Equivalence} readily implies:
\begin{corollary}
\label[corollary]
{ComparisonMapIsRationalEquivalence}
  The comparison map $\hat 1$ \cref{ComparisonMapFromUniversalProperty} is a rational homotopy equivalence.
\end{corollary}
\begin{proof}
  By Serre's finiteness theorem \cref{SerreFiniteness}, the only non-torsion homotopy groups of $S^4$ are in degrees 4 and 7. Therefore, with \cref{ComparisonMapIs7Equivalence} it is sufficient to see that all homotopy groups of $B \widehat{\mathrm{SU}(2)}$ in degree $\geq 8$ are finite. 

  Now the homotopy LES in degrees $n \geq 9$,
  \begin{equation}
    \begin{tikzcd}[
      row sep=small
    ]
      &&
      \overbrace{
      \pi_{n+1}\,
      B^8 \mathbb{Z}
      }^{\smash{0}}
      \ar[
        dll,
        snake left,
        "{
          \partial_{n}
        }"{description}
      ]
      \\
      \pi_n\,
      B \widehat{\mathrm{SU}(2)}
      \ar[
        r,
        "{ \sim }"
      ]
      &
      \pi_n\,
      B \mathrm{SU}(2)
      \ar[r]
      & 
      \underbrace{
      \pi_{n}\,
      B^8 \mathbb{Z}
      }_{\smash{0}}
      \mathrlap{\,,}
    \end{tikzcd}
  \end{equation}
  immediately implies that
  \begin{equation}
    \begin{aligned}
      \pi_n\, 
      B \widehat{\mathrm{SU}(2)}
      & 
      \simeq
      \pi_n\,
      B\mathrm{SU}(2)
      \\
      & \simeq
      \pi_{n-1}\, 
      \mathrm{SU}(2)
      \\
      & \simeq
      \pi_{n-1}\,
      S^3
      \mathrlap{\,,}
    \end{aligned}
  \end{equation}
  which is indeed finite, again by Serre's theorem.

  In the remaining degree $n = 8$ the homotopy LES instead gives
  \begin{equation}
    \begin{tikzcd}[
      row sep=small
    ]
      &&
      \overbrace{
      \pi_{9}\,
      B^8 \mathbb{Z}
      }^{\smash{0}}
      \ar[
        dll,
        snake left,
        "{
          \partial_{8}
        }"{description}
      ]
      \\
      \pi_8\,
      B \widehat{\mathrm{SU}(2)}
      \ar[
        r,
      ]
      &
      \pi_8\,
      B \mathrm{SU}(2)
      \ar[r]
      & 
      \underbrace{
      \pi_{8}\,
      B^8 \mathbb{Z}
      }_{\smash{\mathbb{Z}}}
      \mathrlap{\,.}
    \end{tikzcd}
  \end{equation}
  But since $\pi_8\, B \mathrm{SU}(2) \simeq \pi_7\, S^3$ is again finite, the last map here is zero and hence exactness implies, as before, that also $\pi_8\, B \widehat{\mathrm{SU}}(2) \simeq \pi_7\, S^3$ is finite.
\end{proof}
\begin{corollary}
\label[corollary]
{HatBSU2AdmissibleClassifyingSpace}
  Under rationalization of the comparison map $\hat 1$ \cref{ComparisonMapFromUniversalProperty},
  the minimal Sullivan  model of $B \widehat{\mathrm{SU}(2)}$ \cref{TheSpaceHatBSU2} is identified with that of the 4-sphere:
  \begin{equation}
    \begin{tikzcd}
      \mathfrak{l}S^4
      \ar[
        r,
        "{
          \mathfrak{l}\, \hat 1
        }",
        "{ \sim }"{swap}
      ]
      &
      \mathfrak{l}
      \, B \widehat{\mathrm{SU}(2)}
      \mathrlap{\,,}
    \end{tikzcd}
  \end{equation}
  hence:
  \begin{equation}
  \label
  {MinimalSullivanModelOfHatBSU2}
    \mathrm{CE}\bracket({
      \mathfrak{l}\,
      B \widehat{\mathrm{SU}(2)}
    })
    \simeq
  \mathbb{R}_{_{\mathrm{d}}}
  \left[
    \begin{aligned}
      &g_{4}
      \\[-2pt]
      &g_{7}
    \end{aligned}
  \right]\big/
  \left({
    \begin{aligned}
    \mathrm{d}\, g_{4}
    &= 
    0
    \\[-2pt]
    \mathrm{d}\, g_{7}
    & =
    \tfrac{1}{2}
    g_{4}^2
    \end{aligned}
  }\right)
  \mathrlap{.}
  \end{equation}
\end{corollary}
By comparison with \cref{TheCFieldBianchis}, this shows that $B \widehat{\mathrm{SU}(2)}$ is an admissible classifying space for electromagnetic flux quantization of the gauge sector of 11D supergravity, as discussed in \cref{OnTheChargeQuantization}.

Of course, the minimal Sullivan model \cref{MinimalSullivanModelOfHatBSU2} by itself (disregarding the comparison map $\hat 1$) also follows by more elementary rational homotopy theory, not requiring \cref{ComparisonMapIsRationalEquivalence}. And since rationally $\mathrm{U}(2) \simeq_{_{\mathbb{Q}}} \mathrm{U}(1) \times \mathrm{SU}(2)$, we similarly have:
\begin{lemma}
\label[lemma]
{MinSullivanModelOfBhatU2}
  The minimal Sullivan model of $B \widehat{\mathrm{U}(2)}$ \cref{TheSpaceBHatU2} is:
  \begin{equation}
  \label
  {TheMinSullivanModelOfBhatU2}
    \mathrm{CE}\bracket({
      \mathfrak{l}
      B \widehat{\mathrm{U}(2)}
    })
    \simeq
  \mathbb{R}_{_{\mathrm{d}}}
  \left[
    \begin{aligned}
      & f_2
      \\
      &f_{4}
      \\
      &h_{7}
    \end{aligned}
  \right]\Big/
  \left({
    \begin{aligned}
    \mathrm{d}\, f_{2}
    &= 
    0
    \\[-2pt]
    \mathrm{d}\, f_4
    & = 0
    \\[-2pt]
    \mathrm{d}\, h_7
    & =
    \tfrac{1}{2}
    f_{4}^2
    \end{aligned}
  }\right)
  \mathrlap{.}
  \end{equation}
\end{lemma}
Comparison with \cref{BianchiOnF2F4} hence shows that $B \widehat{\mathrm{U}(2)}$ is an admissible classifying space for that situation, as discussed there.

From combining \cref{TheMinSullivanModelOfBhatU2} with \cref{MinimalSullivanModelOfHatBSU2} we also have:
\begin{lemma}
 The \emph{relative} minimal Sullivan model \cref{RelativeMinimalSullivanModel} of $B \widehat{\mathrm{SU}(2)}$ $\mathrm{U}(1)$-fibered \cref{MapFromBHatSU2ToBhatU2} over $B \widehat{\mathrm{U}(2)}$ is 
  \begin{equation}
    \mathrm{CE}\bracket({
      \mathfrak{l}_{_{
        B \widehat{\mathrm{U}(2)}
      }}
      B \widehat{\mathrm{SU}(2)}
    })
    \simeq
    \mathbb{R}_{_{\mathrm{d}}}
    \left[
    \begin{aligned}
      & \theta_1
      \\
      & f_2
      \\
      & f_4
      \\
      & 
      h_7
    \end{aligned}
    \right]
    \Big/
    \left(
    \begin{aligned}
      \mathrm{d}\,
      \theta_1 & = f_2
      \\
      \mathrm{d}\,
      f_2 & = 0
      \\
      \mathrm{d}\, 
      f_4 & = 0
      \\
      \mathrm{d}\,
      h_7 & = 
      \tfrac{1}{2} f_4^2
    \end{aligned}
    \right)
    \mathrlap{\,,}
  \end{equation}
  whence the identity map on $B \widehat{\mathrm{SU}(2)}$ is also modeled by the quasi-isomorphism:
  \begin{equation}
    \label
    {QuasiIsoModelingIdentityOnBhatSU2}
    \begin{tikzcd}[
      sep=0pt
    ]
      \mathfrak{l}_{_{
        B \widehat{\mathrm{U}(2)}
      }}
      B \widehat{\mathrm{SU}(2)}
      \ar[
        rr,
        "{ \sim }"
      ]
      &&
      \mathfrak{l}\,
      B \widehat{\mathrm{SU}(2)}
      \\
      f_4
      &\longmapsfrom&
      f_4
      \\
      h_7 
        &\longmapsfrom& 
      h_7
      \mathrlap{\,.}
    \end{tikzcd}
  \end{equation}
\end{lemma}

Similarly, \cref{ComparisonOfFibrationsIs7Equivalence} readily implies, furthermore:
\begin{corollary}
\label[corollary]
{FiberComparisonIsRationalEquivalence}
  The fiber comparison map $\hat h_{\mathbb{H}}$ \cref{ComparisonMapFromS7} is a rational homotopy equivalence.
\end{corollary}
\begin{proof}
  By \cref{ComparisonOfFibrationsIs7Equivalence} it is sufficient to observe that all higher homotopy groups $\pi_{\geq 8} S^7$ are finite \cref{SerreFiniteness}.
\end{proof}
In particular, therefore the \emph{relative} minimal Sullivan model \cref{RelativeMinimalSullivanModel} of the fibration $\wp$ \cref{FibrationOfHattedSpaces} coincides with that of the quaternionic Hopf fibration, which is \parencites[Prop. 3.20]{FSS20-H}[Lem. 2.13]{FSS22-Twistorial}:
\begin{equation}
\label
{SullivanModelOfB7ZRelativeToBHatSUTwo}
  \mathrm{CE}\bracket({
    \mathfrak{l}_{_{
      B \widehat{\mathrm{SU}(2)}
    }}
    B^7 \mathbb{Z}
  })
  \simeq
  \mathrm{CE}\bracket({
    \mathfrak{l}_{_{S^4}}
    S^7
  })
  \simeq
  \mathbb{R}_{_{\mathrm{d}}}
  \left[
  \begin{aligned}
    & 
    h_3
    \\[2pt]
    &
    g_4
    \\[-3pt]
    & 
    g_7
  \end{aligned}
  \right]
  \Big/
  \left(
  \begin{aligned}
    \mathrm{d}\, h_3
    & = 
    g_4
    \\[2pt]
    \mathrm{d}\, g_4
    & = 0
    \\[-3pt]
    \mathrm{d}\, g_7 
    & = 
    \tfrac{1}{2} g_4^2
  \end{aligned}
  \right)
  \mathrlap{.}
\end{equation}
Comparison with \cref{Overview11D} therefore shows that both $h_{\mathbb{H}}$ as well as our more K-flavored $\wp$ classify admissible twisted relative flux quantization laws for 11D SuGra with M5-probes.

\paragraph
{Cyclification}

Given the minimal Sullivan model  of a space $\ClassifyingB$,
\begin{equation}
  \mathrm{CE}\bracket({
    \mathfrak{l}
    \ClassifyingB
  })
  \simeq
  \mathbb{R}_{_{\mathrm{d}}}
  \Big[
    (e^i)_{i \in I}
  \Big]
  \big/
  \Big(
    \mathrm{d}_{_{\ClassifyingB}}
    e^i
    =
    P^i(\vec e \,)
  \Big),
\end{equation}
we recall that the minimal Sullivan model of its cyclic loop space, $\mathrm{Cyc}\bracket(\ClassifyingB)$ \cref{Cyclification}, is given (\cite{Burghelea1985}, cf. \parencites[Prop. 3.2]{FSS17-Sphere}[Prop. 2.4(ii)]{SatiVoronov2024}) by:
\begin{equation}
\label
{SullivanModelOfCyclicLoopSpace}
  \mathrm{CE}\bracket({
    \mathfrak{l}
    \mathrm{Cyc}\bracket({
      \ClassifyingB
    })
  })
  \simeq
  \mathbb{R}_{_{\mathrm{d}}}
  \left[
  \begin{aligned}
    &
    (e^i)_{i \in I}
    \\
    & (\mathrm{s}e^i)_{i \in I}
    \\
    & \;\;\omega_2
  \end{aligned}
  \right]
  \Big/
  \left(
  \begin{aligned}
    \mathrm{d}\, 
    e^i
    & =
    \mathrm{d}_{_{\ClassifyingB}}
    e^i
    +
    \omega_2 \, \mathrm{s}e^i
    \\
    \mathrm{d}\,
    \mathrm{s}e^i
    & =
    -\mathrm{s}\bracket({
      \mathrm{d}_{_{\ClassifyingB}}
      e^i
    })
    \\
    \mathrm{d}\, \omega_2 & = 0
  \end{aligned}
  \right)
  \mathrlap{,}
\end{equation}
where $\mathrm{s}e^i$ are degree-shifted generators, $\mathrm{deg}\bracket({\mathrm{s}e^i}) = \mathrm{deg}\bracket({e^i}) - 1$, and on the right we understand the \emph{shift operator}
\begin{equation}
  \mathrm{s} : 
  \left\{
  \begin{aligned}
    e^i & \mapsto \mathrm{s}e^i
    \\
    \mathrm{s}e^i & \mapsto 
    0
  \end{aligned}
  \right.
\end{equation}
as uniquely extended to a $(-1)$-graded derivation on $\mathbb{R}\bracket[{ \bracket({e^i})_{i \in I}, \bracket({\mathrm{s}e^i})_{i \in I} }]$.

Using this with \cref{HatBSU2AdmissibleClassifyingSpace} it follows (as in \cite[\S 3]{FSS17-Sphere}, cf. \cite[Ex. 2.26]{GSS25-TD}) that:
\begin{corollary}
\label[corollary]
{RationalModelOfCyclifiedSpace}
The minimal Sullivan model of the cyclified loop space of  $B \widehat{\mathrm{SU}(2)}$ is:
\begin{equation}
  \mathrm{CE}\bracket({
    \mathfrak{l}
    \mathrm{Cyc}\bracket({
      B
      \widehat{\mathrm{SU}(2)}
    })
  })
  \simeq
  \mathbb{R}_{_{\mathrm{d}}}
  \left[
  \begin{aligned}
    & f_2
    \\
    & f_4
    \\
    & f_6
    \\
    & h_3
    \\
    & h_7
  \end{aligned}
  \right]
  \Big/
  \left(
  \begin{aligned}
    \mathrm{d}\,
    f_2
    & = 0
    \\
    \mathrm{d}\,
    f_4
    & = h_3 f_2
    \\
    \mathrm{d}\,
    f_6
    & = h_3 f_4
    \\
    \mathrm{d}\, h_3
    & = 0
    \\
    \mathrm{d}\, h_7
    & =
    \tfrac{1}{2}
    f_4^2 
    -
    f_2 f_6
  \end{aligned}
  \right)
  \mathrlap{.}
\end{equation}
\end{corollary}
\begin{proof}
  In view of the general formula \cref{SullivanModelOfCyclicLoopSpace} applied to \cref{MinimalSullivanModelOfHatBSU2},
  this is a matter of naming the generators as follows
  \begin{equation}
  \label
  {IdentifyingGeneratorsForModelOfCyc}
    \begin{aligned}
      h_7 & := + g_7
      \\
      f_6 & := - \mathrm{s}g_7
      \\
      f_4 & := + g_4
      \\
      h_3 & := + \mathrm{s}f_4
    \end{aligned}
  \end{equation}
  (noting the conventional choice of relative signs).
\end{proof}
Comparison with \cref{TheActualBianchisIn10D} thus shows that $\mathrm{Cyc}\bracket({B \widehat{\mathrm{SU}(2)}})$ is therefore indeed the classifying space of an admissible flux quantization law for the gauge sector of (non-massive) type IIA supergravity.

\begin{corollary}
\label[corollary]
{PullbackAlongExtCycAdjunct}
  Pullback along the map 
  $
    \begin{tikzcd}
      \mathfrak{l}\,
        B \widehat{\mathrm{U}(2)}
      \ar[r]
      &
      \mathfrak{l}\,
      \mathrm{Cyc}\bracket({
        B \widehat{\mathrm{SU}(2)}
      })
    \end{tikzcd}
  $
  which is induced on minimal Sullivan models by the map \cref{TheAdjunctMapIntoCycBHatSU2,TheAdjunctMapIntoCycBHatSU2Derived},
  is the identity of $f_2$, $f_4$ and $h_7$ and vanishes on $f_6$ and $h_3$.
\end{corollary}
\begin{proof}
By construction in \cref{TheExtCycAdjunct}, the map is the $\mathrm{Ext}/\mathrm{Cyc}$-adjunct of the identity on $\mathfrak{l}\,B \widehat{\mathrm{SU}(2)}$. Applying the general formula \parencites[Prop. 2.25]{GSS25-TD}[Thm. 3.8]{FSS18-TD} for these adjuncts to the above model \cref{QuasiIsoModelingIdentityOnBhatSU2} for that identity map yields the correspondence:
  \begin{equation}
    \begin{tikzcd}[
      sep=0pt
    ]
      \mathfrak{l}_{_{
        B \widehat{\mathrm{U}(2)}
      }}
      B \widehat{\mathrm{SU}(2)}
      \ar[
        rr,
        "{ \sim }"
      ]
      &&
      \mathfrak{l}\,
      B \widehat{\mathrm{SU}(2)}
      \\
      f_4
      &\longmapsfrom&
      g_4
      \\
      h_7 
        &\longmapsfrom& 
      g_7
    \end{tikzcd}
    \;\;\;
    \leftrightsquigarrow
    \;\;\;
    \begin{tikzcd}[
      sep=0pt
    ]
      \mathfrak{l}B \widehat{\mathrm{U}(2)}
      \ar[
        rr
      ]
      &&
      \mathfrak{l}\mathrm{Cyc}\bracket({
        B \widehat{\mathrm{SU}(2)}
      })
      \\
      f_2 
        &\longmapsfrom& 
      f_2
      \\
      f_4 
        &\longmapsfrom& 
      f_4
      \\
      0 
        &\longmapsfrom& 
      f_6
      \\
      0 
        &\longmapsfrom& 
      h_3
      \\
      h_7 
        &\longmapsfrom& 
      h_7
      \mathrlap{\,.}
    \end{tikzcd}
  \end{equation}
  On the right this is the claimed pullback.
\end{proof}

Further with \cref{FiberComparisonIsRationalEquivalence}, the immediate generalization of this computation to the relative minimal Sullivan model of $\wp$  yields, as in \cite[\S 3]{Banerjee2026M5brane}:
\begin{corollary}
The relative minimal Sullivan model of the cyclification of $\wp$ is:
\begin{equation}
  \mathrm{CE}\bracket({
    \mathfrak{l}_{_{
      \mathrm{Cyc}\scaledbracket({
        B
        \widehat{\mathrm{SU}(2)}
      })
    }}
    \mathrm{Cyc}\bracket({
      B^7 \mathbb{Z}
    })
  })
  \simeq
  \mathbb{R}_{_{\mathrm{d}}}
  \left[
  \begin{aligned}
    & \mathscr{f}_2
    \\
    & \mathscr{h}_3
    \\[5pt]
    & f_2
    \\
    & f_4
    \\
    & f_6
    \\
    & h_3
    \\
    & h_7
  \end{aligned}
  \right]
  \Big/
  \left(
  \begin{aligned}
    \mathrm{d}\, 
    \mathscr{f}_2
    & =
    h_3
    \\
    \mathrm{d}\,
    \mathscr{h}_3
    & =
    f_4 
      - 
    f_2 \mathscr{f}_2    
    \\[5pt]
    \mathrm{d}\,
    f_2
    & = 0
    \\
    \mathrm{d}\,
    f_4
    & = h_3 f_2
    \\
    \mathrm{d}\,
    f_6
    & = h_3 f_4
    \\
    \mathrm{d}\, h_3
    & = 0
    \\
    \mathrm{d}\, h_7
    & =
    \tfrac{1}{2}
    f_4 f_4 
    -
    f_2 f_6
  \end{aligned}
  \right)
  \mathrlap{.}
\end{equation}
\end{corollary}
\begin{proof}
  In view of the general formula \cref{SullivanModelOfCyclicLoopSpace}, now applied to \cref{SullivanModelOfB7ZRelativeToBHatSUTwo}, this is a matter of naming the further generators, beyond \cref{IdentifyingGeneratorsForModelOfCyc},
  as follows:
  \begin{equation}
    \begin{aligned}
      \mathscr{h}_3
      & := 
      + h_3
      \\
      \mathscr{f}_2 
        & := - \mathrm{s}\mathscr{h}_3
    \end{aligned}
  \end{equation}
  (noting again the conventional choice of signs).
\end{proof}
Comparison with \cref{Overview} shows that, therefore, $\wp$ is indeed an admissible flux quantization law for type IIA with D4-probes.

\newpage
\section
{Conclusion}

The global completion of a higher gauge theory by flux/charge quantization is a choice, but candidate quantization laws in nontrivial situations have hardly been explored.

The traditional proposal that RR-flux (D-brane charge) is quantized in twisted topological K-theory suffers (besides from the problematic ontology of the ``massive'' flux species $F_0$ and $F_{10}$) from its unjustified disregard of the nonlinear Bianchi identity for the dual NS-flux. At the same time, such nonlinear Bianchis do not admit flux quantization in the traditionally considered abelian (stable) cohomology theories like K-theory.

The solution is to pass to flux quantization in nonabelian cohomology theories. For 11D Sugra, the minimal choice (``Hypothesis H'') of full electromagnetic flux quantization (minimal in number of cells in the corresponding classifying space) is in 4-Cohomotopy, classified by the 4-sphere. Under the process of double dimensional reduction of flux-quantization laws, namely by cyclification of their classifying spaces, this implies a consistent quantization of (actual, non-massive) type IIA SuGra.

However, the resulting \emph{cyclified 4-Cohomotopy} is in general not closely related to K-theory, at least not manifestly so. Hence, in order to see to which extent the K-theory conjecture may be salvaged while retaining dimensional oxidation to 11D (hence retaining M/IIA duality even after flux quantization), here we looked for a more ``K-flavored'' but still strictly admissible choice of electromagnetic flux quantization of the gauge sector of type IIA SuGra. 

One such choice is naturally suggested by understanding the classifying space $S^4 \simeq \mathbb{H}P^1$ as the quaternionic projective line, and as such as the 4-skeleton of the infinite quaternionic projective space $\mathbb{H}P^\infty \simeq B \mathrm{SU}(2)$. The latter, being $\mathrm{U}(1)$-fibered over a finite stage in the classifying space for complex K-theory, is closer to K-theory, and as such might be understood as classifying a form of ``unstable K-theory'', at least bringing in cocycle data in the form of unitary bundles. 

Our main mathematical observation here (\cref{ComparisonMapIs7Equivalence}) is that, up to dimension 7, the homotopy theoretic identification of $S^4$ with $B \mathrm{SU}(2)$ is obstructed exactly only by the square of the second Chern class, which vanishes on $S^4$ but not on $B \mathrm{SU}(2)$: Forcing this squared class to vanish yields a deformed classifying space $B \widehat{\mathrm{SU}(2)}$ which does classify an admissible electromagnetic flux quantization of 11D SuGra but is more ``K-theory flavored'' than $S^4$ (with which however it agrees up to dimension 7), in that unitary bundles, at least of rank 2, serve as cocycle data. 

Concretely, the double-dimensional reduction of this quantization law yields an admissible proper quantization of the IIA flux densities classified by the cyclic loop space $\mathrm{Cyc}\bracket({B \widehat{\mathrm{SU}(2)}})$, which receives a canonical map from the ``deformed unstable K-theory'' classified by $B \widehat{\mathrm{U}(2)}$. 

This map thus exhibits $\mathrm{Cyc}\bracket({B \widehat{\mathrm{SU}(2)}})$ as classifying a nonabelian cohomology theory that is an unstable form (correction) of the traditionally expected twisted K-theory, while being strictly admissible as a flux-quantization law for (non-massive) type IIA supergravity, including the nonlinear Bianchi identity of the dual NS-flux. A variant also extends to the twisted relative quantization of the further worldvolume fluxes on D4-branes.

Moreover, still by the main theorem (\cref{ComparisonMapIs7Equivalence}), the flux quantization of this ``K-flavored'' law strictly coincides with the cohomotopical one (Hypothesis H) on 10D spacetimes of the form $\mathbb{R}^{1,3} \times X^6$. Hence, it immediately inherits the latter's favorable properties regarding topological conditions expected in M-theory.

In this sense, $\mathrm{Cyc}\bracket({B \widehat{\mathrm{SU}(2)}})$ classifies a nonabelian cohomology theory that fares fairly well in meeting both: (i) the strict technical demands on a proper flux quantization of type IIA SuGra, as well as (ii) further string/M-theoretic folklore expectations.

But it needs to be re-emphasized that global completion of any higher gauge theory, like type IIA SuGra, is a choice that requires attention, and that there are an infinitude of further admissible flux quantization laws waiting to be explored and to be compared to given desiderata and expectations.

\medskip

\noindent \textbf{Acknowledgments.}
We thank Alonso Perez-Lona for bringing further references 
to our attention. 


\printbibliography

\end{document}